\documentclass[11pt]{article}
\usepackage[margin=1in]{geometry}

\usepackage{amsmath,amsthm,amsfonts,amssymb}
\usepackage{graphicx}
\usepackage{url}
\usepackage{lineno}
\usepackage{hyperref}

\usepackage{xcolor}

\title{Homotopic curve shortening and the affine curve-shortening flow\footnote{A preliminary version of this article appeared in \emph{Proc. 36th Int. Symp. Computational Geometry (SoCG'20)}, article 12, 15 pages, 2020.}}
\date{}
\author{Sergey Avvakumov\thanks{\texttt{savvakumov@gmail.com}. Department of Mathematical Sciences, University of Copenhagen, Copenhagen, Denmark. Has received funding from the Austrian Science Fund (FWF), Project P31312-N35, and the European Research Council under the European Union's Seventh Framework Programme ERC Grant agreement ERC StG 716424 - CASe.} \and Gabriel Nivasch\thanks{ \texttt{gabrieln@ariel.ac.il}. Ariel University, Ariel, Israel.}}

\newcommand{\R}{\mathbb R}
\newcommand{\Z}{\mathbb Z}

\newcommand{\eps}{\varepsilon}
\newcommand{\So}{\mathbb S^1}
\newcommand{\infl}{\imath}                   
\newcommand{\cros}{\chi}                      
\newcommand{\tP}{t_P}                         
\newcommand{\tOne}{t_1}                         
\newcommand{\cg}{c_{\mathrm g}}               
\newcommand{\crand}{c_{\mathrm r}}               
\DeclareMathOperator{\HCS}{\mathrm HCS}
\DeclareMathOperator{\snap}{\mathrm snap}

\newtheorem{theorem}{Theorem}
\newtheorem{lemma}[theorem]{Lemma}
\newtheorem{observation}[theorem]{Observation}

\newtheorem{conjecture}[theorem]{Conjecture}

\makeatletter
\newtheorem*{rep@theorem}{\rep@title}
\newcommand{\newreptheorem}[2]{%
\newenvironment{rep#1}[1]{%
 \def\rep@title{#2 \ref{##1}}%
 \begin{rep@theorem}}%
 {\end{rep@theorem}}}
\makeatother

\newreptheorem{theorem}{Theorem}
\newreptheorem{lemma}{Lemma}

\begin{document}
\maketitle

\begin{abstract}
We define and study a discrete process that generalizes the convex-layer decomposition of a planar point set. Our process, which we call \emph{homotopic curve shortening} (HCS), starts with a closed curve (which might self-intersect) in the presence of a set $P\subset \R^2$ of point obstacles, and evolves in discrete steps, where each step consists of (1) taking shortcuts around the obstacles, and (2) reducing the curve to its shortest homotopic equivalent.

We find experimentally that, if the initial curve is held fixed and $P$ is chosen to be either a very fine regular grid or a uniformly random point set, then HCS behaves at the limit like the affine curve-shortening flow (ACSF). This connection between ACSF and HCS generalizes the link between ACSF and convex-layer decomposition (Eppstein et al., 2017; Calder and Smart, 2020), which is restricted to convex curves.

We prove that HCS satisfies some properties analogous to those of ACSF: HCS is invariant under affine transformations, preserves convexity, and does not increase the total absolute curvature. Furthermore, the number of self-intersections of a curve, or intersections between two curves (appropriately defined), does not increase. Finally, if the initial curve is simple, then the number of inflection points (appropriately defined) does not increase.
\end{abstract}

\section{Introduction}

Let $\So$ be the unit circle. In this paper we call a continuous function $\gamma:[0,1]\to\R^2$ a \emph{path}, and a continuous function $\gamma:\So\to\R^2$ a \emph{closed curve}, or simply a \emph{curve}. If $\gamma$ is injective then the curve or path is said to be \emph{simple}. We say that two paths or curves $\gamma$, $\delta$ are \emph{$\eps$-close} to each other if their Fr\'echet distance is at most $\eps$, i.e.~if they can be re-parametrized such that for every $t$, the Euclidean distance between the points $\gamma(t)$, $\delta(t)$ is at most $\eps$.

\subsection{Shortest homotopic curves}

Let $P$ be a finite set of points in the plane, which we regard as obstacles. Two curves $\gamma, \delta$ that avoid $P$ are said to be \emph{homotopic} if there exists a way to continuously transform $\gamma$ into $\delta$ while avoiding $P$ at all times. And two paths $\gamma, \delta$ that avoid $P$ (except possibly at the endpoints) and satisfy $\gamma(0)=\delta(0)$, $\gamma(1)=\delta(1)$ are said to be \emph{homotopic} if there exists a way to continuously transform $\gamma$ into $\delta$, without moving their endpoints, while avoiding $P$ at all times (except possibly at the endpoints). We extend these definitions to the case where $\gamma$ avoids obstacles but $\delta$ does not, by requiring the continuous transformation of $\gamma$ into $\delta$ to avoid obstacles at all times except possibly at the last moment.

For every curve (resp.~path) $\gamma$ in the presence of obstacles there exists a unique shortest curve (resp.~path) $\delta$ that is homotopic to $\gamma$. The problem of computing the shortest path or curve homotopic to a given piecewise-linear path or curve, under the presence of polygonal or point obstacles, has been studied extensively. A simple and efficient algorithm for this task is the so-called ``funnel algorithm'' \cite{chazelle82,hs94,lp84,lm85}. See also~\cite{bespa03,Cabello2004,ekl06}.

\subsection{The affine curve-shortening flow}

In the \emph{affine curve-shortening flow}, a smooth curve
$\gamma\subset \R^2$ varies with time in the following way. At each
moment in time, each point of $\gamma$ moves perpendicularly to the
curve, towards its local center of curvature, with instantaneous
velocity $r^{-1/3}$, where $r$ is that point's radius of curvature at
that time. See Figure~\ref{fig:A:C:S:F}.

\begin{figure}
\centerline{\includegraphics[width=3in]{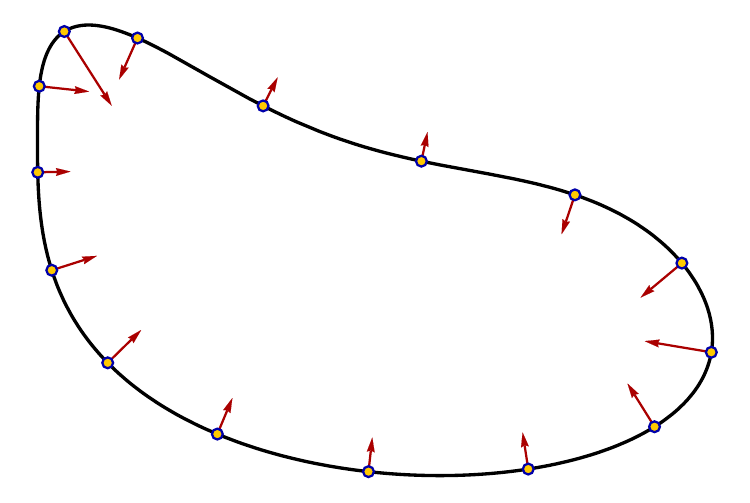}}
    \caption{Affine curve-shortening flow. The arrows indicate the instantaneous velocity of different points along the curve at the shown time moment.}
    \label{fig:A:C:S:F}%
\end{figure}

The ACSF was first studied by Alvarez et al.~\cite{aglm-afeip-93} and
Sapiro and Tannenbaum~\cite{st-aiss-93}.  It differs from the more
usual \emph{curve-shortening flow} (CSF)~\cite{c-gceip-03,cz-csp-01},
in which each point is given instantaneous velocity $r^{-1}$.

Unlike the CSF, the ACSF is invariant under affine transformations:
Applying an affine transformation to a curve, and then performing the
ACSF, gives the same results (after rescaling the time parameter
appropriately) as performing the ACSF and then applying the affine
transformation to the shortened curves. Moreover, if the affine
transformation preserves area, then the time scale is unaffected.

The ACSF was originally applied in computer vision, as a way of
smoothing object boundaries~\cite{c-gceip-03} and of computing shape
descriptors that are insensitive to the distortions caused by changes
of viewpoint.

\paragraph{Properties of the CSF and ACSF for simple curves.} Under either the CSF or the ACSF, a simple curve remains simple, and its length decreases strictly with time (\cite{cz-csp-01}, \cite{st-aiss-93}, resp.). Furthermore, a pair of disjoint curves, run simultaneously, remain disjoint at all times (\cite{white02}, \cite{ast-ahenc-98}, resp.). More generally, the number of intersections between two curves never increases (\cite{ang_blowup91}, \cite{ast-ahenc-98}, resp.).

Given a smooth closed curve $\gamma:[0,1]\to \R^2$, let $\alpha(s)\in\So$ be the unit vector tangent to $\gamma(s)$ for each $s\in[0,1]$. Then the \emph{total absolute curvature} of $\gamma$ is the total distance traversed by $\alpha(s)$ in $\So$ as $s$ goes from $0$ to $1$. If $\gamma$ is convex then its total absolute curvature is exactly $2\pi$; otherwise, it is larger than $2\pi$. Under either the CSF or the ACSF, the total absolute curvature of a curve decreases strictly with time and tends to $2\pi$ (\cite{gage1986,grayson1987}, \cite{ast-ahenc-98}, resp.). Moreover, the number of inflection points of a simple curve does not increase with time (\cite{ang_blowup91}, \cite{ast-ahenc-98}, resp.).

Under the CSF, a simple curve eventually becomes convex and then converges to a circle as it collapses to a point~\cite{gage1986,grayson1987}. Correspondingly, under the ACSF, a simple curve becomes convex and then converges to an ellipse as it collapses to a point~\cite{ast-ahenc-98}.

\paragraph{Self-intersecting curves.} When the initial curve is not simple, a self-intersection might collapse and form a cusp with infinite curvature. For the CSF, it has been shown that, as long as the initial curve satisfies some natural conditions, it is possible with some care to define the flow past the singularity~\cite{altschuler1992,ang_blowup91}. Angenent~\cite{ang_blowup91} generalized these results to a wide range of flows, but unfortunately the ACSF is not included in this range~\cite{ast-ahenc-98}. Hence, no rigorous results have been obtained for self-intersecting curves under the ACSF. Still, ACSF computer simulations can be run on curves that have self-intersections or singularities with little difficulty.

\subsection{Relation to the convex-layer decomposition}

Let $P$ be a finite set of points in the plane. The \emph{convex-layer decomposition} (also called the \emph{onion decomposition}) of $P$ is the partition of $P$ into sets $P_1, P_2, P_3, \ldots$ obtained as follows: Let $Q_0 = P$. Then, for each $i\ge 1$ for which $Q_{i-1}\neq \emptyset$, let $P_i$ be the set of vertices of the convex hull of $Q_{i-1}$, and let $Q_i = Q_{i-1} \setminus P_i$. In other words, we repeatedly remove from $P$ the set of vertices of its convex hull. See~\cite{b-omd-76,c-clps-85,d-co-04,e-chp-82}.

Eppstein et al.~\cite{gp_acsf}, following Har-Peled and Lidick\'y~\cite{hl-pg-13}, studied \emph{grid peeling}, which is the convex-layer decomposition of subsets of the integer grid $\Z^2$. Eppstein et al.~found an experimental connection between ACSF for convex curves and grid peeling. Specifically, let $\gamma$ be a fixed convex curve. Let $n$ be large, let $(\Z/n)^2$ be the uniform grid with spacing $1/n$, and let $P_n(\gamma)$ be the set of points of $(\Z/n)^2$ that are contained in the region bounded by $\gamma$. Then, as $n\to\infty$, the convex-layer decomposition of $P_n(\gamma)$ seems experimentally to converge to the ACSF evolution of $\gamma$, after the time scale is adjusted appropriately. They formulated this connection precisely in the form of a conjecture. They also raised the question whether there is a way to generalize the grid peeling process so as to approximate ACSF for non-convex curves as well.

Dalal~\cite{d-co-04} studied the convex-layer decomposition of point sets chosen uniformly and independently at random from a fixed convex domain, in the plane as well as in $\R^d$. Recently, Calder and Smart~\cite{caldersmart} proved that the above-described correspondence between ACSF and the convex-layer decomposition holds if, instead of a uniform grid, one uses a random point set sampled uniformly and independently within the region bounded by the convex curve $\gamma$.

\subsection{Our Contribution}

In this paper we describe a generalization of the convex-layer decomposition to non-convex, and even non-simple, curves. We call our process \emph{homotopic curve shortening}, or HCS. Under HCS, an initial curve evolves in discrete steps in the presence of point obstacles. We find that, if the obstacles form a uniform grid, then HCS shares the same experimental connection to ACSF that grid peeling does. Hence, HCS is the desired generalization sought by Eppstein et al.~\cite{gp_acsf}. We also find that the same experimental connection between ACSF and HCS (and in particular, between ACSF and the convex-layer decomposition) holds when the obstacles are distributed uniformly at random, with the sole difference being in the constant of proportionality.

Although the experimental connection between HCS and ACSF seems hard to prove, we do prove that HCS satisfies some simple properties analogous to those of ACSF: HCS is invariant under affine transformations, preserves convexity, and does not increase the total absolute curvature. Furthermore, the number of self-intersections of a curve, or intersections between two curves (appropriately defined), does not increase. Finally, if the initial curve is simple, then the number of inflection points (appropriately defined) does not increase.

\paragraph{Organization of This Paper.} In Section~\ref{sec_HCS} we describe homotopic curve shortening (HCS), our generalization of the convex-layer decomposition. In Section~\ref{sec_experiments} we present our conjectured connection between ACSF and HCS, as well as experimental evidence supporting this connection. In Section~\ref{sec_properties} we state our theoretical results, to the effect that HCS satisfies some properties analogous to those of ACSF. In Section~\ref{sec_proofs} we prove the results stated in Section~\ref{sec_properties}. Appendices~\ref{app_unique}--\ref{app_implementation} include proofs of some known results for the sake of completeness, as well as implementation details of our experiments.

\section{Homotopic curve shortening}\label{sec_HCS}

Let $P$ be a finite set of obstacle points. A \emph{$P$-curve} (resp.~\emph{$P$-path}) is a curve (resp.~path) that is composed of straight-line segments, where each segment starts and ends at obstacle points.

\emph{Homotopic curve shortening} (HCS) is a discrete process that starts with an initial $P$-curve $\gamma_0$ (which might self-intersect), and at each step, the current $P$-curve $\gamma_n$ is turned into a new $P$-curve $\gamma_{n+1} = \HCS_P(\gamma_n)$.

The definition of $\gamma'=\HCS_P(\gamma)$ for a given $P$-curve $\gamma$ is as follows. Let $(p_0, \ldots, p_{m-1})$ be the circular list of obstacle points visited by $\gamma$. Call $p_i$ \emph{nailed} if $\gamma$ goes straight through $p_i$, i.e.~if $\angle p_{i-1} p_i p_{i+1} = \pi$.\footnote{All indices in circular sequences are modulo the length of the sequence.} Let $(q_0, \ldots, q_{k-1})$ be the circular list of nailed vertices of $\gamma$. Suppose first that $k\ge 1$. Then $\gamma'$ is obtained through the following three substeps:
\begin{enumerate}
\item \emph{Splitting.} We split $\gamma$ into $k$ $P$-paths $\delta_0, \ldots, \delta_{k-1}$ at the nailed vertices, where each  $\delta_i$ goes from $q_i$ to $q_{i+1}$.
\item \emph{Shortcutting.} For each non-endpoint vertex $p_i$ of each $\delta_i$, we make the curve avoid $p_i$ by taking a small shortcut. Specifically, let $\eps>0$ be sufficiently small, and for each $i$ let $C_{p_i}$ be a circle of radius $\eps$ centered at $p_i$. Let $e_i$ be the segment $p_{i-1}p_i$ of $\delta_i$. Let $x_i=e_i\cap C_{p_i}$ and $y_i=e_{i+1}\cap C_{p_i}$. Then we make the path go straight from $x_i$ to $y_i$ instead of through $p_i$. Call the resulting path $\rho_i$, and let $\rho$ be the curve obtained by concatenating all the paths $\rho_i$.
\item \emph{Shortening.} Each $\rho_i$ in $\rho$ is replaced by the shortest $P$-path homotopic to it. The resulting curve is $\gamma'$.
\end{enumerate}
If $\gamma$ has no nailed vertices ($k=0$) then $\gamma'$ is obtained by performing the shortcutting and shortening steps on the single closed curve $\gamma$. Note that if $P$ is in general position (no three points on a line) then there will never be nailed vertices. Figure~\ref{fig_HCS_step} illustrates one HCS step on a sample curve.

\begin{figure}
\centerline{\includegraphics{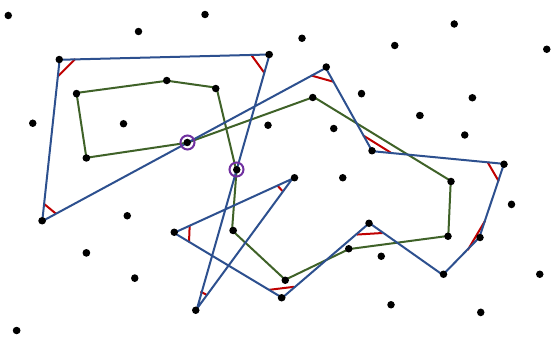}}
\caption{\label{fig_HCS_step}Computation of a single step of homotopic curve shortening: Given a $P$-curve $\gamma$ (blue), we first identify its nailed vertices (purple). In this case, the two nailed vertices split $\gamma$ into two paths $\delta_0$, $\delta_1$. In each $\delta_i$ we take a small shortcut around each intermediate vertex (red). Then we replace each $\delta_i$ by the shortest path homotopic to it, obtaining the new $P$-curve $\gamma'=\HCS_P(\gamma)$ (green).}
\end{figure}

The process terminates when the curve collapses to a point. This will certainly happen after a finite number of steps, since at each step the curve gets strictly shorter, and there is a finite number of distinct $P$-curves of at most a certain length.

\paragraph{HCS for convex curves.} If the initial curve $\gamma_0$ is the boundary of the convex hull of $P$, then the HCS evolution of $\gamma_0$ is equivalent to the convex-layer decomposition of $P$. Namely, for every $i\ge 0$, the curve $\gamma_i$ is the boundary of a convex polygon, and the set of vertices of this polygon equals the $(i+1)$-st convex layer of $P$. See Section~\ref{sec_properties} below.

\section{Experimental connection between ACSF and HCS}\label{sec_experiments}

Our experiments show that HCS, using $P=(\Z/n)^2$ as the obstacle set, approximates ACSF at the limit as $n\to\infty$, just as grid peeling approximates ACSF for convex curves. The connection between the two processes is formalized in the following conjecture, which generalizes Conjecture 1 of~\cite{gp_acsf}.

\begin{conjecture}\label{conj_grid}
There exists a constant $\cg\approx 1.6$ such that the following is true: Let $\delta$ be an initial curve. Fix a time $t>0$ for which $\delta'=\delta(t)$ under ACSF is defined. For a fixed $n$, let $\gamma_0$ be the shortest curve homotopic to $\delta$ under obstacle set $P_n = (\Z/n)^2$. Let $m=\lfloor\cg tn^{4/3}\rfloor$, and let $\gamma_m=\HCS_P^{(m)}(\gamma_0)$ be the result of $m$ iterations of HCS starting with $\gamma_0$. Then, as $n\to\infty$, the Fr\'echet distance between $\gamma_m$ and $\delta'$ tends to $0$.
\end{conjecture}

Furthermore, we find that the connection between ACSF and HCS also holds if the uniform grid $(\Z/n)^2$ is replaced by a random point set, though with a different constant of time proportionality.

\begin{conjecture}\label{conj_rand}
There exists a constant $\crand\approx 1.3$ such that the following is true: Let $\delta$ be an initial curve, contained in a convex region $R$ of area $A$. Fix a time $t>0$ for which $\delta'=\delta(t)$ under ACSF is defined. For a fixed $n$, let $P$ be a set of $An^2$ obstacle points chosen uniformly and independently at random from $R$. Let $\gamma_0$ be the shortest curve homotopic to $\delta$ under obstacle set $P$. Let $m=\lfloor\crand tn^{4/3}\rfloor$, and let $\gamma_m=\HCS_P^{(m)}(\gamma_0)$ be the result of $m$ iterations of HCS starting with $\gamma_0$. Then, as $n\to\infty$, the Fr\'echet distance between $\gamma_m$ and $\delta'$ is almost surely smaller than $\eps$, for some $\eps = \eps(n)$ that tends to $0$ with $n$.
\end{conjecture}

As mentioned above, Conjecture~\ref{conj_rand} was recently proven by Calder and Smart~\cite{caldersmart} for the special case where $\delta$ is a convex curve. Furthermore, they point out that their experiments seem to suggest $\crand=4/3$.

\subsection{Experiments}\label{subsec_experiments}

We tested Conjectures~\ref{conj_grid} and~\ref{conj_rand} on a variety of test curves. We found that for all our test curves, the result of HCS does seem to converge to the result of ACSF as $n\to\infty$, both for grid and for random obstacle sets.

\begin{figure}
\centerline{\includegraphics{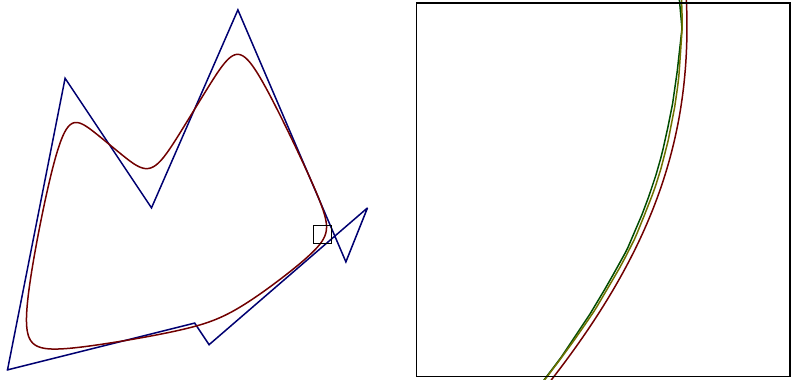}}
\caption{\label{fig_camelfish}Left: Initial curve $\Delta$ (blue) and simulated ACSF result after the curve's length reduced to 70\% of its original length (red). Right: Comparison between ACSF approximation (red), HCS with $n=10^7$ uniform-grid obstacles (green), and HCS with $n=10^7$ random obstacles (yellow) on a small portion of the curve.}
\end{figure}

Let us illustrate our experiments on the piecewise-linear curve $\Delta$ having vertices $(0,0),\allowbreak (0.16,0.81),\allowbreak (0.4,0.45),\allowbreak(0.64,1),\allowbreak(0.94,0.3),\allowbreak(1,0.45),\allowbreak(0.56,0.07),\allowbreak(0.52,0.13)$. We approximated ACSF using an approach similar to the one in \cite{gp_acsf}. We ran our ACSF simulation on $\Delta$ until we obtained a curve $\Delta'$ whose length equals 70\% of the original length of $\Delta$.  See Figure~\ref{fig_camelfish} (left). This happened at $t^*\approx 0.0266$. By this time, the self-intersection and an inflection point of the curve have disappeared.

Then we introduced in the unit square $[0,1]^2\supset\Delta$ a set $P$ of $n$ obstacle points, where $P$ is either a uniform grid (i.e.~a $\sqrt n\times\sqrt n$ grid) $G_n$, or a random set $R_n$. For each case, we initially snapped each vertex of $\Delta$ to its closest point in $P$, obtaining a $P$-curve, and then we ran HCS until the length of the curve shrank to 70\% of its original length, obtaining a new curve $\Delta''=\Delta''(P)$. We did this for several values of $n$. For each case, we computed $h(\Delta',\Delta'')$, where $h(\gamma_1,\gamma_2)$ for piecewise-linear curves $\gamma_1, \gamma_2$ is defined as the maximum distance between a vertex of one curve and the closest point on the other curve. (For ``nice'' curves as ours, there is no significant difference, if at all, between this distance $h$ and either the Hausdorff or the Fr\'echet distance between the two curves.)

For random obstacles, we conducted this experiment for $n=10^4, 10^5, 10^6, 10^7$, taking the average of $5$ samples for each value of $n$. Our random-obstacle program is limited by memory rather than by time, since it stores all the obstacle points in memory. For uniform-grid obstacles, we conducted this experiment also for $n=10^8$. After this point, our ACSF approximation $\Delta'$ does not seem to be accurate enough for reliable comparisons. The results are shown in Figure~\ref{fig_hplot} (left).

\begin{figure}
\centerline{\includegraphics{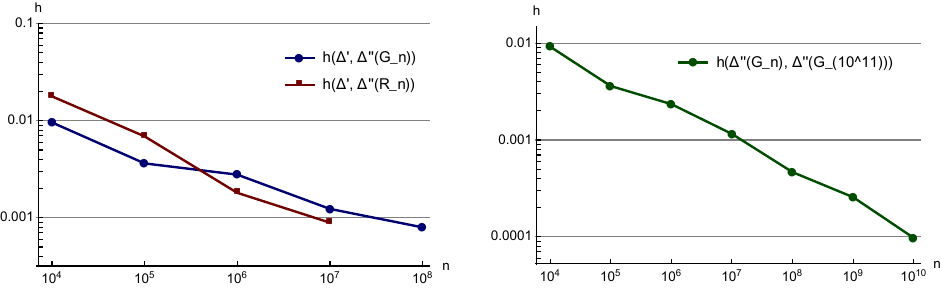}}
\caption{\label{fig_hplot}Left: Distance between ACSF approximation and HCS with uniform-grid obstacles (blue curve) or random obstacles (red curve, average of $5$ trials), for increasing values of $n$, the number of obstacles. Right: Distance between HCS with uniform-grid obstacles for $n=10^4, \ldots, 10^{10}$ and with $n=10^{11}$.}
\end{figure}

We also checked whether the relation between the ACSF time $t^*$ and the number of HCS iterations $m$ behaves as predicted by Conjectures~\ref{conj_grid} and~\ref{conj_rand}. For this purpose, we computed $c=m/(t^*n^{2/3})$ for each case, and checked whether $c$ is roughly constant. The results are shown in Table~\ref{table_c}.

\begin{table}
\centerline{\begin{tabular}{c|cc|cc}
$n$&iterations with $G_n$&$\cg$&avg.~iterations with $R_n$&$\crand$\\ \hline
$10^4$&$20$&$1.616$&$15.6$&$1.261$\\
$10^5$&$93$&$1.619$&$75.2$&$1.309$\\
$10^6$&$434$&$1.628$&$351.2$&$1.317$\\
$10^7$&$2006$&$1.621$&$1628.6$&$1.316$\\
$10^8$&$9266$&$1.613$
\end{tabular}}
\caption{\label{table_c}Approximations of the constants $\cg$ and $\crand$ given by the experiments.}
\end{table}

As we can see, Conjectures~\ref{conj_grid} and~\ref{conj_rand} are well supported by the experiments.

Finally, we measured the rate of convergence of the uniform-grid HCS to its limit shape as $n\to\infty$. To this end, we computed $h(\Delta''(G_n),\Delta''(G_m))$ for $n\in\{10^4, 10^5, \ldots, 10^{10}\}$ and $m=10^{11}$. See Figure~\ref{fig_hplot} (right). As we can see, increasing $n$ by a factor of $10$ has the effect of multiplying the distance by roughly a factor of $0.47$.

See Appendix~\ref{app_implementation} for some implementation details of our ACSF and HCS simulations.

\section{Properties of homotopic curve shortening}\label{sec_properties}

In this paper we prove that HCS satisfies some properties analogous to those of ACSF. This section contains the statements of our results, and the proofs appear in Section~\ref{sec_proofs}.

\begin{lemma}\label{thm_affine}
HCS is invariant under affine transformations. Namely, if $P$ is a set of obstacle points, $\gamma$ is a $P$-curve, and $T$ is a non-degenerate affine transformation, then $T(\HCS_P(\gamma))=\HCS_{T(P)}(T(\gamma))$.
\end{lemma}

In particular, if $T$ is a grid-preserving affine transformation, meaning that $T$ maps $(\Z/n)^2$ injectively to itself, then the HCS evolution using $P=(\Z/n)^2$ (as in Conjecture~\ref{conj_grid}) is unaffected by $T$. Hence, HCS on uniform-grid obstacles is invariant under a certain subset of the area-preserving affine transformations, just as in grid peeling~\cite{gp_acsf}. 

Also, if $T$ is an area-preserving affine transformation, then the probability distribution of random sets $P$ in the convex region $R$ of Conjecture~\ref{conj_rand} stays unaffected after applying $T$ to $R$.

\begin{lemma}\label{thm_convex}
Let $\gamma$ be a simple $P$-curve, and let $\gamma'=\HCS_P(\gamma)$. If $\gamma$ is the boundary of a convex polygon, then so is $\gamma'$. Hence, under HCS, once a curve becomes the boundary of a convex polygon, it stays that way.
\end{lemma}

The \emph{total absolute curvature} of a piecewise-linear curve $\gamma$ with vertices $(p_0, \ldots, p_{m-1})$ is the sum of the exterior angles $\sum_{i=0}^{m-1} (\pi-|\angle p_{i-1}p_ip_{i+1}|)$. It equals $2\pi$ if $\gamma$ is the boundary of a convex polygon, and it is larger than $2\pi$ otherwise.

\begin{lemma}\label{thm_curvature}
Let $\gamma$ be a $P$-curve, and let $\gamma'=\HCS_P(\gamma)$. Let $\alpha, \alpha'$ be the total absolute curvature of $\gamma, \gamma'$, respectively. Then $\alpha \ge \alpha'$. Hence, under HCS, the total absolute curvature of a curve never increases.
\end{lemma}

If $\gamma, \delta$ are disjoint $P$-curves, then $\HCS_P(\gamma), \HCS_P(\delta)$ are not necessarily disjoint. Similarly, if $\gamma$ is a simple $P$-curve, then $\HCS_P(\gamma)$ is not necessarily simple. See Figure~\ref{fig_not_disjoint}.

\begin{figure}
\centerline{\includegraphics{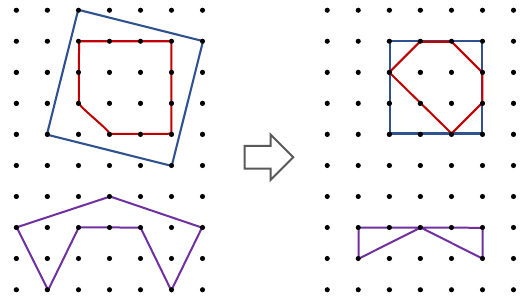}}
\caption{\label{fig_not_disjoint}HCS might cause disjoint curves to intersect, or a simple curve to self-intersect.}
\end{figure}

Curves $\gamma, \delta$ are called \emph{disjoinable} if they can be made into disjoint curves by peforming on them an arbitrarily small perturbation. Similarly, a curve $\gamma$ is called \emph{self-disjoinable} if it can be turned into a simple curve by an arbitrarily small perturbation. Note that if $\gamma$ is self-disjoinable then $\gamma$,~$\gamma$ are disjoinable, though the reverse is not necessarily true: Consider for example a curve $\gamma$ that makes two complete clockwise turns around the unit circle.

A self-disjoinable curve is also called \emph{weakly simple}. Akitaya et al.~\cite{weaklysimple} recently found an algorithm for recognizing weakly simple polygons in time $O(n\log n)$. (See also~\cite{fulek_nlogn} for an $O(n\log n)$-time algorithm for the more general problem of recognizing \emph{weak-embeddings} of graphs.) 

An intersection between two curves, or between two portions of one curve, is called \emph{transversal}, if at the point of intersection both curves are differentiable and their normal vectors are not parallel at that point. If all intersections between curves $\gamma_1$ and $\gamma_2$ are transversal, then we say that $\gamma_1, \gamma_2$ are themselves \emph{transversal}. Similarly, if all self-intersections of $\gamma$ are transversal, then we say that $\gamma$ is \emph{self-transversal}. (Transversal and self-transversal curves are sometimes called \emph{generic}, see e.g.~\cite{untangling}.)

If $\gamma$ is self-transversal, we denote by $\cros(\gamma)$ the number of self-intersections of $\gamma$.\footnote{A self-intersection in a curve $\gamma:\So\to\R^2$ is a pair $s\neq t$ such that $\gamma(s)=\gamma(t)$. Hence, if $\gamma$ passes $k$ times through a certain point, that counts as $\binom{k}{2}$ self-intersections.} If $\gamma$ is not self-transversal, then we define $\cros(\gamma)$ as the minimum of $\cros(\widehat\gamma)$ among all self-transversal curves $\widehat\gamma$ that are $\eps$-close to $\gamma$, for all small enough $\eps>0$. Hence, $\cros(\gamma)=0$ if and only if $\gamma$ is self-disjoinable. We define similarly the number of intersections $\cros(\gamma_1,\gamma_2)$ between two curves. Then, $\gamma_1$ and $\gamma_2$ are disjoinable if and only if $\cros(\gamma_1,\gamma_2)=0$. Fulek and T\'oth recently proved that the problem of computing $\cros(\gamma)$ is NP-hard~\cite{fulek_nphard}.

\begin{theorem}\label{thm_crossings}
Let $\gamma$ be a $P$-curve, and let $\gamma'=\HCS_P(\gamma)$. Then their self-intersection numbers satisfy $\cros(\gamma')\le \cros(\gamma)$. Let $\delta$ be another $P$-curve, and let $\delta'=\HCS_P(\delta)$. Then their intersection numbers satisfy $\cros(\gamma',\delta')\le\cros(\gamma,\delta)$. In particular, if $\gamma$ is self-disjoinable, so is $\gamma'$, and if $\gamma,\delta$ are disjoinable, then so are $\gamma',\delta'$. Hence, under HCS, the intersection and self-intersection numbers never increase.
\end{theorem}

With the technique of Theorem~\ref{thm_crossings} we can obtain an upper bound on the number of iterations of HCS:

\begin{theorem}\label{thm_bound_by_peeling}
For a fixed obstacle set $P$, the $P$-curve $\gamma$ that maximizes the number of HCS iterations is the boundary of the convex hull of $P$.
\end{theorem}

As mentioned above, if $\gamma$ is the boundary of the convex hull of $P$, then the HCS process starting with $\gamma$ is just the convex-layer decomposition of $P$. Hence, by Theorem~\ref{thm_bound_by_peeling}, if $|P|=n$ then the HCS process starting with any $P$-curve ends in at most $n/2$ iterations. If $P = \{1,2,\ldots,\sqrt{n}\}^2$ then the process ends in at most $O(n^{2/3})$ iterations (by Har-Peled and Lidick\'y~\cite{hl-pg-13}). And if $P$ is uniformly and independently chosen at random inside a fixed convex domain, then the expected number of iterations is $O(n^{2/3})$ (by Dalal~\cite{d-co-04}).

Recall that an obstacle set $P$ is in general position if no three points of $P$ lie on a line; and recall that if $P$ is in general position then there are no nailed vertices in HCS.

\begin{theorem}\label{thm_gp}
Let $P$ be an obstacle set in general position. Let $\gamma$ be a simple $P$-curve. Then $\HCS_P(\gamma)$ is also simple. Let $\gamma_1, \gamma_2$ be disjoint $P$-curves. Then $\HCS_P(\gamma_1), \HCS_P(\gamma_2)$ are also disjoint. Hence, under HCS with obstacles in general position, a simple curve stays simple, and a pair of disjoint curves stay disjoint.
\end{theorem}

Let $\gamma$ be a simple piecewise-linear curve with vertices $(v_0, \ldots, v_{n-1})$. Assume that the sequence of vertices is minimal, meaning no $v_{i-1},v_i,v_{i+1}$ lie on a straight line. An \emph{inflection edge} of $\gamma$ is an edge $v_iv_{i+1}$ such that the previous and next vertices $v_{i-1}, v_{i+2}$ lie on opposite sides of the line through $v_i, v_{i+1}$. Let $\infl(\gamma)$ be the number of inflection edges of $\gamma$. Note that $\infl(\gamma)$ is always even, since every inflection edge lies either after a sequence of clockwise vertices and before a sequence of counterclockwise vertices, or vice versa.

If $\gamma$ is not simple but self-disjoinable, then we define $\infl(\gamma)$ as the minimum of $\infl(\gamma')$ over all simple piecewise-linear curves $\gamma'$ that are $\eps$-close to $\gamma$, for all sufficiently small $\eps>0$. (Note that for a given $\gamma$ there might exist different curves $\gamma'$ with different values of $\infl(\gamma')$. For example, if $\gamma$ goes from a point $p$ to a point $q$ and back $n$ times, then $\gamma'$ could be a spiral with just two inflection edges, or a double zig-zag with $2n-2$ inflection edges.)

\begin{theorem}\label{thm_infl}
Let $\gamma$ be self-disjoinable, and let $\gamma'=\HCS_P(\gamma)$. Then their inflection-edge numbers satisfy $\infl(\gamma')\le \infl(\gamma)$. Hence, under HCS on a self-disjoinable curve, the curve's number of inflection edges never increases.
\end{theorem}

\section{Proofs}\label{sec_proofs}

In order to prove the properties stated in Section~\ref{sec_properties}, we rely on two different approaches for computing shortest homotopic curves. The first approach is the one given by \cite{chazelle82,hs94,lp84,lm85}, which uses a triangulation of the ambient space. The second approach consists of repeatedly releasing unstable vertices. We start by describing these two approaches in detail.

For simplicity, we restrict ourselves to piecewise-linear curves. By the relative simplicial approximation theorem (Zeeman~\cite{zeeman64}), we can restrict ourselves to piecewise-linear homotopies as well.

\subsection{Triangulations}\label{subsec_triang}

Let $P$ be a finite set of point obstacles, and let $\gamma$ be a piecewise-linear curve avoiding $P$. Assume without loss of generality that $\gamma$ is contained in the convex hull of $P$ (by adding points outside the convex hull of $\gamma$ if necessary). Let $\mathcal T$ be a triangulation of the convex hull of $P$ using the points of $P$ as vertices.

\begin{figure}
\centerline{\includegraphics{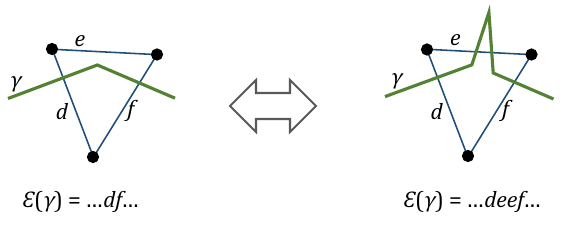}}
\caption{\label{fig_add_remove_ee} Combinatorial changes in the edge sequence of a continuously changing curve $\gamma$.}
\end{figure}

We can assume without loss of generality that the curve $\gamma$ intersects each triangle edge transversally. Let $\mathcal E=\mathcal E(\gamma)$ be the circular sequence of triangle edges intersected by $\gamma$. Then a piecewise-linear homotopic change of $\gamma$ can only have two possible types of effects on $\mathcal E$: Either an adjacent pair $ee$ is inserted somewhere in the sequence, or an existing such pair is deleted. See Figure~\ref{fig_add_remove_ee}. Hence, two curves $\gamma$, $\gamma'$ are homotopic if and only if their corresponding edge sequences $\mathcal E(\gamma)$, $\mathcal E(\gamma')$ are \emph{equivalent}, in the sense that they can be transformed into one another by a sequence of operations of these two types.

Call an edge sequence $\mathcal E(\gamma)$ \emph{reduced} if it contains no adjacent pair $ee$. Then every edge sequence is equivalent to a unique reduced sequence. (Proof sketch: Supposing for a contradiction that there exist two distinct equivalent reduced sequences $S_1, S_2$, consider a transformation of $S_1$ into $S_2$ that uses the minimum possible number of deletions, and among those, consider one in which the first deletion is done as early as possible. Then it is easy to arrive at a contradiction.)

Hence, in order to compute the shortest curve homotopic to $\gamma$, we first compute $\mathcal E(\gamma)$, then we reduce this sequence by repeatedly removing adjacent pairs, obtaining a reduced sequence $\mathcal E'$, then we place a point $x(e)$ on each $e\in\mathcal E'$, and then we slide the points $x(e)$ along their edges so as to minimize the length of the curve. This last step can be done using the ``funnel algorithm'' (see \cite{chazelle82,hs94,lp84,lm85}); we omit the details.

\begin{figure}
\centerline{\includegraphics{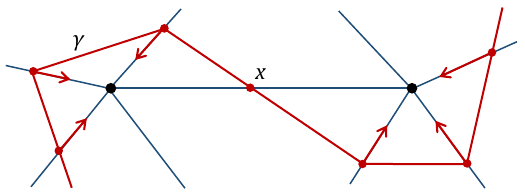}}
\caption{\label{fig_pt_not_unique}In the shortest curve homotopic to $\gamma$, the position of the point $x$ is not uniquely defined.}
\end{figure}

For the sake of completeness, we include in Appendix~\ref{app_unique} a proof that there is always a unique shortest curve. Note that, even though the shortest curve is always unique, the final positions of the points $x(e)$ are not necessarily unique. This can happen if a triangulation edge is an edge of the final curve. See Figure~\ref{fig_pt_not_unique}.

\subsection{The vertex release algorithm}\label{sec_vrelease}

We now present another simple algorithm for the shortest homotopic curve problem. This algorithm is not mentioned in any previous publication that we are aware of, but it is similar in spirit to well-known algorithms, in particular to the funnel algorithm.

As a warm-up, let us first consider the case in which the obstacles are are not single points but rather polygons. Let $\gamma$ be a piecewise-linear curve that avoids the interior of all the obstacles. Call a vertex $v=\gamma(t_i)$ of $\gamma$ \emph{unstable} if $v$ does not lie on any obstacle, or if $v$ lies on the boundary of an obstacle $T$, but $T$ lies locally on the side of $\gamma$ at which the angle is larger than $\pi$. If $v$ is unstable, then the process of \emph{releasing} $v$ is as follows: Let $u=\gamma(t_{i-1})$ and $w=\gamma(t_{i+1})$ be the previous and next vertices of $\gamma$. Suppose first that $u\neq w$. Let $\Delta$ be the triangle $uvw$, and let $S$ be the set of obstacle vertices that lie inside $\Delta$. Let $u, z_1, \ldots, z_k, w$ be the vertices of the convex hull of $S\setminus\{v\}$ in order. Then we replace $v$ by $z_1, \ldots, z_k$ in $\gamma$. The new vertices $z_1, \ldots, z_k$ are necessarily stable, but $u$ and $w$ might change from stable to unstable or vice versa. If $u=w$ then we simply remove $v$ and $w$ from $\gamma$. See Figure~\ref{fig_release} (left).

\begin{figure}
\centerline{\includegraphics{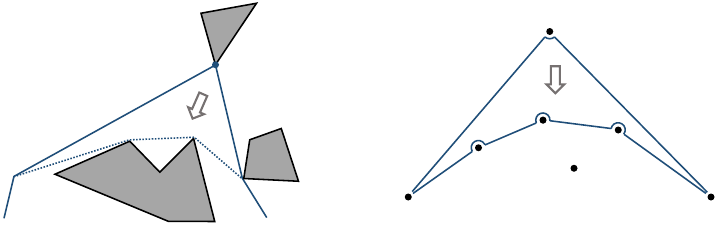}}
\caption{\label{fig_release}Releasing an unstable vertex in the presence of polygonal obstacles (left) or point obstacles (right).}
\end{figure}

Then the algorithm consists of releasing unstable vertices one by one, in an arbitrary order, until no more unstable vertices remain.

If there are also point obstacles, then the algorithm becomes slightly more complicated. For each curve vertex $v=\gamma(t_i)$ that lies on a point obstacle, we need to remember the corresponding signed angle $\alpha_i$ that the curve turns around the obstacle, since this angle could be larger than $2\pi$ in absolute value. The angle $\alpha_i$ is always congruent modulo $2\pi$ to $\angle uvw$, where $u=\gamma(t_{i-1})$ and $w=\gamma(t_{i+1})$ are the previous and next vertices. The vertex $v$ is unstable if and only if $|\alpha_i|<\pi$.

If $v$ is unstable, then in order to release it we proceed as described above, and we update the angles as follows (see Figure~\ref{fig_release}, right):
\begin{itemize}
\item If $u\neq w$ then we give to each new vertex $z_j$ the unique appropriate angle $\beta_j$ that has the opposite sign of $\alpha_i$ and satisfies $\pi\le|\beta_j|<2\pi$. We then update the angles $\alpha_{i-1}$ and $\alpha_{i+1}$ at $u$ and $w$ as follows: Denote $z_0=u$ and $z_{k+1} = w$ (in order to handle properly the case $k=0$). We add to $\alpha_{i-1}$ the angle $\angle vuz_1$, and we add to $\alpha_{i+1}$ the angle $\angle z_kwv$.
\item If $u=w$ then we update $\alpha_{i-1}$ by adding to it the angle $\alpha_{i+1}$.
\end{itemize}

\begin{lemma}\label{lemma_alg}
The vertex release algorithm always finds the shortest curve $\delta$ homotopic to the given piecewise-linear curve $\gamma$, irrespective of the order of release of the vertices.
\end{lemma}

\begin{proof}
The vertex-release algorithm ends in a finite number of iterations, since at each iteration the length of the curve strictly decreases and the number of vertices not touching obstacles never increases. Hence, there exists a finite number of distinct possible curves the algorithm might go through.

The final curve $\delta$ is clearly homotopic to the initial curve $\gamma$. Moreover, the final curve is locally shortest, in the sense that for every sufficiently small $\eps>0$, every other curve $\delta'$ homotopic to $\gamma$ and $\eps$-close to $\delta$ is longer than $\delta$. But in the Euclidean plane with polygonal or point obstacles, for every curve there exists a unique locally shortest curve homotopic to it. (For completeness, we include the proof of this fact in Appendix~\ref{app_unique}.) Hence, $\delta$ is indeed the shortest curve homotopic to $\gamma$.
\end{proof}

\subsection{Proof of Lemmas~\ref{thm_affine}--\ref{thm_curvature}}

Lemmas~\ref{thm_affine}--\ref{thm_curvature} follow easily from the vertex-release algorithm (Lemma~\ref{lemma_alg}).

\begin{replemma}{thm_affine}
HCS is invariant under affine transformations. Namely, if $P$ is a set of obstacle points, $\gamma$ is a $P$-curve, and $T$ is a non-degenerate affine transformation, then $T(\HCS_P(\gamma))=\HCS_{T(P)}(T(\gamma))$.
\end{replemma}

\begin{proof}
The claim follows from the fact that shortest homotopic curves and paths are invariant under affine transformations. Namely, let $\gamma$ be a curve or path in the presence of obstacle points $P$, let $\delta$ be the shortest curve or path homotopic to $\gamma$, and let $T:\R^2\to\R^2$ be a non-degenerate affine transformation. Then the shortest curve or path homotopic to $T(\gamma)$ in the presence of $T(P)$ is $T(\delta)$. This, in turn, follows from the fact that $T$ does not affect whether a vertex is stable or unstable, and furthermore, if a vertex is unstable, then it does not matter whether we first release the vertex and then apply $T$, or do these operations in the opposite order.
\end{proof}

\begin{replemma}{thm_convex}
Let $\gamma$ be a simple $P$-curve, and let $\gamma'=\HCS_P(\gamma)$. If $\gamma$ is the boundary of a convex polygon, then so is $\gamma'$. Hence, under HCS, once a curve becomes the boundary of a convex polygon, it stays that way.
\end{replemma}

\begin{proof} 
This follows from the fact that the property of being the boundary of a convex polygon is preserved by each vertex release.
\end{proof}

\begin{replemma}{thm_curvature}
Let $\gamma$ be a $P$-curve, and let $\gamma'=\HCS_P(\gamma)$. Let $\alpha, \alpha'$ be the total absolute curvature of $\gamma, \gamma'$, respectively. Then $\alpha \ge \alpha'$. Hence, under HCS, the total absolute curvature of a curve never increases.
\end{replemma}

\begin{proof}
Given a piecewise-linear curve $\gamma$ with vertices $(p_0, \ldots, p_{m-1})$, let $v_i\in\So$ be the unit vector parallel to $\overrightarrow{p_i p_{i+1}}$ for each $i$ . Call a tour of $\So$ \emph{valid} if it visits the vectors $v_0, v_1, \ldots, v_{m-1}, v_0$ in this order. Then the total absolute curvature of $\gamma$ equals the length of the shortest valid tour of $\So$.

Now let $\gamma$ be a given $P$-curve, and let $\gamma'=\HCS_P(\gamma)$. Recall that $\gamma'$ is obtained from $\gamma$ by a series of vertex releases. Each vertex release replaces two adjacent vectors $v_i, v_{i+1}\in \So$ by a certain number $k\ge 1$ of vectors $w_1, \ldots, w_k$ lying between them, in this order. Hence, the shortest valid tour of $\So$ for the old vector sequence goes from $v_i$ to $v_{i+1}$ through $w_1, \ldots, w_k$, and hence this tour is also valid for the new vector sequence.
\end{proof}

\subsection{Proof of Theorems~\ref{thm_crossings}--\ref{thm_gp}}\label{sec_crossings}

The proof of Theorems~\ref{thm_crossings}--\ref{thm_gp} is based on the triangulation technique. Let $\gamma$ be a $P$-curve, let $\eps>0$ be small enough, and let $\widehat\gamma$ be a self-transversal curve that is $\eps$-close to $\gamma$ and has the minimum possible number of self-intersections. In order to prove Theorem~\ref{thm_crossings} regarding intersection and self-intersection numbers, we proceed as follows:
\begin{enumerate}
\item We show that, without loss of generality, we can assume that $\widehat\gamma$ passes through the ``correct side'' of each non-nailed obstacle, as in the ``shortcutting'' step of HCS.
\item We modify $\widehat\gamma$ homotopically, by first eliminating repetitions in its edge sequence $\mathcal E$ and then sliding its vertices along the triangulation edges, until each vertex comes within $\eps$ of its final position as given by $\gamma'=\HCS_P(\gamma)$. We show that the number of self-intersections never increases in the process.
\end{enumerate}
The case of two curves is similar.

In order to do the first step, we define a type of curves that are $\eps$-close to $P$-curves and pass through the ``correct side'' of non-nailed obstacles. We call them \emph{type-2 curves}. We also define a ``snapping'' operation, which transforms $\widehat\gamma$ into a type-2 curve without increasing its number of self-intersections.

\begin{figure}
\centerline{\includegraphics{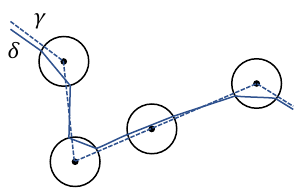}}
\caption{\label{fig_type2}A $P$-curve $\gamma$ and a corresponding type-2 curve $\delta$.}
\end{figure}

\paragraph{Type-2 curves.} Let $\gamma$ be a $P$-curve, let $(p_0, \ldots, p_{k-1})$ be the circular list of obstacles visited by $\gamma$, and let $\eps>0$ be small enough. For each $p\in P$, let $C_p$ be a circle of radius $\eps$ centered at $p$. For each $i$, let $x_i\in C_{p_i}$ be a point at distance at most $\eps^2$ from the segment $p_{i-1}p_i$, and let $y_i\in C_{p_i}$ be a point at distance at most $\eps^2$ from the segment $p_ip_{i+1}$. Then a type-2 curve $\delta$ corresponding to $\gamma$ travels in a straight line from $y_{i-1}$ to $x_i$ and then in a straight line from $x_i$ to $y_i$ for each $i$. See Figure~\ref{fig_type2}. We call each segment $y_{i-1}x_i$ a \emph{long part} and each segment $x_iy_i$ a \emph{short part}. If $\angle p_{i-1}p_ip_{i+1} \neq \pi$ and $\eps$ is chosen small enough, then $p_i$ lies on the side of the curve at which the angle is larger than $\pi$. If $\angle p_{i-1}p_ip_{i+1} = \pi$ then the corresponding short part passes within distance $\eps^2$ of $v_i$.

\paragraph{The snapping operation.} Let $\gamma$ be a $P$-curve, and let $\widehat\gamma$ be a curve $(\eps^2)$-close to $\gamma$. We define the type-2 curve $\snap(\widehat\gamma)$ as follows. For each $p_i$ visited by $\widehat\gamma$ there exists a point $z_i$ in $\widehat\gamma$ that is within distance $\eps^2$ of $p_i$. Let $y_i$ be the first intersection of $\widehat\gamma$ with $C_{p_i}$ that comes after $z_i$, and let $x_i$ be the last intersection of $\widehat\gamma$ with $C_{p_i}$ that comes before $z_i$. (Thus, the part of $\widehat\gamma$ between $x_i$ and $y_i$ is entirely contained in the disk bounded by $C_{p_i}$.) Then we let $\snap(\widehat\gamma)$ be the type-2 curve that uses these points $x_i$, $y_i$ for all $i$ as vertices.

\begin{lemma}\label{lemma_snap}
Let $\gamma$ be a $P$-curve, and let $\widehat\gamma$ be a self-transversal curve $(\eps^2)$-close to $\gamma$, such that no self-intersection of $\widehat\gamma$ occurs on any circle $C_p$. Then the curve $\delta=\snap(\widehat\gamma)$ also is also self-transversal, and it satisfies $\cros(\delta)\le \cros(\widehat\gamma)$.

Similarly, let $\gamma_1, \gamma_2$ be $P$-curves, and let $\widehat\gamma_1$, $\widehat\gamma_2$ be transversal curves $(\eps^2)$-close to them, respectively, such that no intersection between $\widehat\gamma_1$ and $\widehat\gamma_2$ occurs on any circle $C_p$. Then the curves $\delta_1=\snap(\widehat\gamma_1),\delta_2=\snap(\widehat\gamma_2)$ are transversal and satisfy $\cros(\delta_1,\delta_2)\le\cros(\widehat\gamma_1,\widehat\gamma_2)$.
\end{lemma}

\begin{figure}
\centerline{\includegraphics{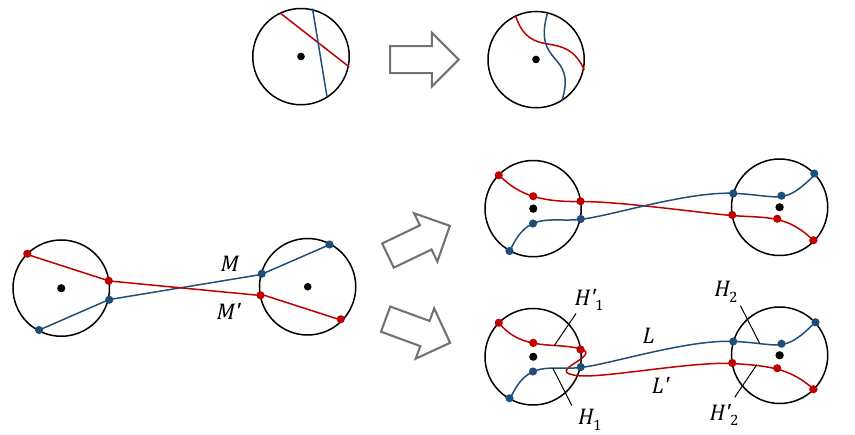}}
\caption{\label{fig_trace_back}Each self-intersection of $\delta$ can be traced back to a self-intersection of $\widehat\gamma$. Different portions of $\delta$ and $\widehat\gamma$ are shown in different colors.}
\end{figure}

\begin{proof}
We will show that each self-intersection of $\delta$ can be mapped to a self-intersection of $\widehat\gamma$, such that different self-intersections of $\delta$ are mapped to different self-intersections of $\widehat\gamma$.

The points $x_i$, $y_i$ define a partition of $\widehat\gamma$ into \emph{long} and \emph{short parts} corresponding to the long and short parts of $\delta$. On each short part of $\widehat\gamma$ inside a circle $C_{p_i}$, pick a point $z_i$ that is at distance at most $2\eps^2$ from $p_i$, and such that $\widehat\gamma$ passes through $z_i$ only once. Each point $z_i$ partitions its short part into two \emph{half-short parts}.

Now, suppose two short parts of $\delta$ within the same circle $C_{p_i}$ intersect. Then the two corresponding short parts of $\widehat\gamma$ must also intersect. See Figure~\ref{fig_trace_back} (top).

Next, suppose two long parts of $\delta$ intersect. Then either they connect two different pairs of circles, or they both connect a circle $C_{p}$ to a circle $C_{q}$. In the first case, trivially the two corresponding long parts of $\widehat\gamma$ also intersect. In the second case, if the two corresponding long parts of $\widehat\gamma$ intersect then we are done. So suppose this is not the case. We claim that in this case, one of the long parts of $\widehat\gamma$ must intersect one of the half-short parts adjacent to the other long part. See Figure~\ref{fig_trace_back} (bottom).

Let $L, L'$ be the two long parts of $\widehat\gamma$. Let $H_1$, $H_2$ be the half-short parts adjacent to $L$, and let $H'_1, H'_2$ be the half-short parts adjacent to $L'$. Let $V$ be the concatenation of $H_1, L, H_2$, and let $V'$ be the concatenation of $H'_1, L', H'_2$. Let $M$, $M'$ be the long parts of $\delta$ corresponding to $L, L'$. Let $W$ be the concatenation of $H_1, M, H_2$, and let $W'$ be the concatenation of $H'_1, M', H'_2$. Since $\delta$ is $(2\eps^2)$-close to $\widehat\gamma$, there exists a homotopy from $V$ to $W$ that does not pass through any endpoint of $V'$, and there exists a homotopy from $V'$ to $W'$ that does not pass through any endpoint of $V$. Therefore, the number of intersections between $V$ and $V'$ has the same parity as the number of intersections between $W$ and $W'$. Since $L$ and $L'$ do not intersect but $M$ and $M'$ intersect once, a single intersection is created in the transition from $V, V'$ to $W, W'$. Hence, an odd number of intersections must have been lost. These intersections must have been between a long part of one curve and a half-short part of the other curve, as claimed.

This finishes the proof that $\cros(\delta)\le\cros(\widehat\gamma)$. A similar argument applies for the case of two curves.
\end{proof}

\begin{reptheorem}{thm_crossings}
Let $\gamma$ be a $P$-curve, and let $\gamma'=\HCS_P(\gamma)$. Then their self-intersection numbers satisfy $\cros(\gamma')\le \cros(\gamma)$. Let $\delta$ be another $P$-curve, and let $\delta'=\HCS_P(\delta)$. Then their intersection numbers satisfy $\cros(\gamma',\delta')\le\cros(\gamma,\delta)$. In particular, if $\gamma$ is self-disjoinable, so is $\gamma'$, and if $\gamma,\delta$ are disjoinable, then so are $\gamma',\delta'$. Hence, under HCS, the intersection and self-intersection numbers never increase.
\end{reptheorem}

\begin{proof}
Let $\gamma$ be a $P$-curve, let $\eps>0$ be small enough, and let $\widehat\gamma$ be a self-transversal curve that is $(\eps^2)$-close to $\gamma$ and has the minimum possible number of self-intersections. Fix a triangulation $\mathcal T$ of $P$. Assume without loss of generality that $\widehat\gamma$ does not pass through any obstacle, and that no self-intersection of $\gamma$ lies on any edge of $\mathcal T$. Let $\eta=\snap(\widehat\gamma)$. Hence, $\eta$ is a self-transveral curve that is $\eps$-close to $\gamma$, passes through the ``correct side'' of each non-nailed obstacle, and does not have more self-intersections than $\widehat\gamma$. Partition $\eta$ into paths $\eta_0, \ldots, \eta_{k-1}$ that are $\eps$-close to the corresponding paths $\delta_0, \ldots, \delta_{k-1}$ of the HCS ``splitting'' step, by introducing split points as follows: For each nailed visit to an obstacle $p\in P$, we choose a split point that is within distance $O(\eps)$ of $p$ and lies on a triangle edge (where the implicit constant depends only on $P$).

\begin{figure}
\centerline{\includegraphics{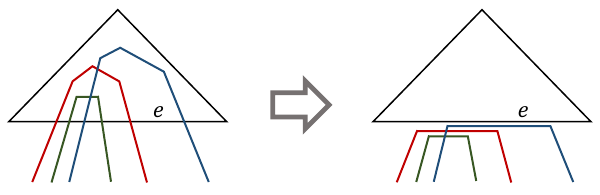}}
\caption{\label{fig_remove_ee}Reducing a curve's edge sequence without increasing its number of self-intersections. Different portions of the curve are shown in different colors.}
\end{figure}

Then we modify each $\eta_i$ into a homotopic path $\eta'_i$ whose edge sequence $\mathcal E(\eta'_i)$ is reduced. We do this without increasing the number of intersections, by repeatedly doing the following: Let $e$ be triangulation edge such that $ee$ appears one or more times in the sequences $\mathcal E(\eta'_i)$. We shortcut the corresponding paths $\eta'_i$ so as to not cross $e$ at all, instead keeping a small distance from $e$. We make the distance to $e$ inversely related to the distance between the two crossing points of $\eta'_i$ with $e$. See Figure~\ref{fig_remove_ee}.

Next, we modify each $\eta'_i$ into $\eta''_i$ by straightening out each part within each triangle of $\mathcal T$. Hence, each $\eta''_i$ is determined by the position of its vertices $x(e)$ along the triangle edges $e$.

Let $\eta''$ be the curve formed by concatenating the paths $\eta''_i$ for all $i$. By appropriately sliding the non-endpoint vertices of the paths $\eta''_i$ along their triangulation edges, we can bring $\eta''$ to be $\eps$-close to $\gamma'$. However, we must be careful to perform this sliding without unnecessarily switching the order of any pair over vertices along the same edge. Meaning, whenever $\gamma'$ contains several vertices along an edge that coincide, we place the corresponding vertices of $\eta''$ within $\eps$ of each other, conserving the order they had before the sliding. Let $\eta'''$ be the curve obtained after sliding the vertices this way. Hence, $\eta'''$ is $\eps$-close to $\gamma'$. Table~\ref{table_curves} summarizes the construction of the different curves.

\begin{table}
\centerline{\begin{tabular}{c|l}
Curve&Description\\\hline
$\gamma$&Given $P$-curve\\
$\widehat\gamma$&Self-transversal curve $(\eps^2)$-close to $\gamma$\\
$\eta$&$\snap(\widehat\gamma)$\\
$\eta'$&Result of reducing edge sequences of paths $\eta_i$ in $\eta$\\
$\eta''$&Result of straightening $\eta'$ within triangles of $\mathcal T$\\
$\eta'''$&Result of sliding vertices of $\eta''$ along $\mathcal T$-edges; $\eps$-close to $\gamma'$\\
$\gamma'$&$\HCS(\gamma)$
\end{tabular}}
\caption{\label{table_curves}Curves used in the proofs of Theorems~\ref{thm_crossings}--\ref{thm_gp}.}
\end{table}

We now claim that the number of self-intersections of $\eta'''$ is not larger than that of $\eta''$.

Let $(x_0, \ldots, x_{m-1})$ be the circular list of vertices of $\eta''$, and let $(y_0, \ldots, y_{m-1})$ be the corresponding vertices of $\eta'''$, meaning that for each $j$, both vertices $x_j, y_j$ lie on the same triangulation edge $e_j$. Consider a self-intersection $z$ lying in some triangle $T\in\mathcal T$, such that $z$ exists in $\eta'''$ but not in $\eta''$. In other words, we have $z = y_j y_{j+1} \cap y_k y_{k+1}$ for some $j,k$, whereas $x_j x_{j+1}\cap x_k x_{k+1}=\emptyset$. This means that one pair of vertices, say $y_{j+1}, y_{k+1}$, lie on the same edge $e_{j+1}=e_{k+1}$ of $T$ and switched their order in the transition from $\eta''$ to $\eta'''$, while the other two vertices $y_j,y_k$ lie on the two other edges of $T$, or else they both lie on the same edge of $T$ but did not switch their order. Let $\ell\ge 1$ be the unique integer such that for all $1\le \ell'\le \ell$, the vertices $y_{j+\ell'}, y_{k+\ell'}$ lie on the same edge $e_{j+\ell'}=e_{k+\ell'}$ and switched order, while this is not true of $y_{j+\ell+1}, y_{k+\ell+1}$. Note that no pair of vertices $(y_{j+\ell'}, y_{k+\ell'})$, $1\le\ell'\le \ell$ can be $\eps$-close to each other, because then we would not have switched their order when transforming $\eta''$ into $\eta'''$.

For each index $m'$, let $s_{m'}, t_{m'}$ be the segments $s_{m'} = x_{m'}x_{m'+1}$ and $t_{m'}=y_{m'}y_{m'+1}$. Then for all $1\le \ell'<\ell$, either none or both of the self-intersections $s_{j+\ell'}\cap s_{k+\ell'}$, $t_{j+\ell'}\cap t_{k+\ell'}$ exist. In contrast, exactly one of the self-intersections $s_{j+\ell}\cap s_{k+\ell}$, $t_{j+\ell}\cap t_{k+\ell}$ exists. Meaning, we have associated the self-intersection $z$ that was \emph{created} to another self-intersection $z'$ that was either created or \emph{destroyed}. See Figure~\ref{fig_created_destroyed} (top). Furthermore, this association is one-to-one, since, if $z'$ was also created, then it is associated back to $z$. We now show that $z'$ must have been destroyed, not created.

\begin{figure}
\centerline{\includegraphics{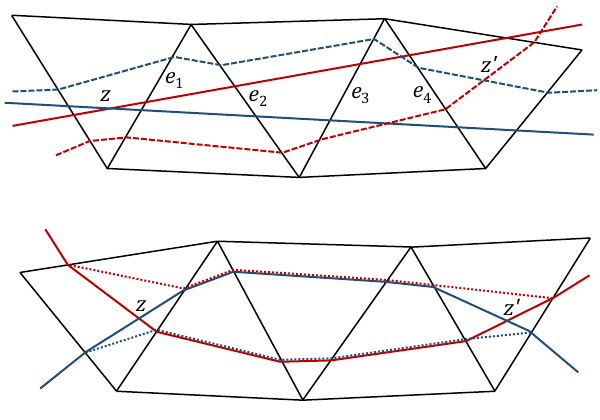}}
\caption{\label{fig_created_destroyed}Top: In the transition from $\eta''$ (dashed curves) to $\eta'''$ (solid curves), vertices at $e_1,e_2,e_3,e_4$ had their order swapped, so self-intersection $z$ was created while self-intersection $z'$ was destroyed. Bottom: If both $z$ and $z'$ were created, then interchanging the position of the intermediate vertices eliminates both self-intersections, and also produces a shorter curve.}
\end{figure}

Suppose for a contradiction that the self-intersection $z'$ was also created. We will use a mix-and-match argument in order to arrive at a contradiction to the fact that $\eta'''$ is within $\eps$ of being homotopically shortest. The mix-and-match consists of interchanging the positions of $y_{j+\ell'}$ and $y_{k+\ell'}$ for all $1\le \ell'\le \ell$ (see Figure~\ref{fig_created_destroyed}, bottom). But in order for the mix-and-match argument to be valid, we need to show that the portion of $\eta'''$ between $y_j$ and $y_{j+\ell}$ is parametrically disjoint from the portion between $y_{k}$ and $y_{k+\ell}$, or in other words, the indices $j, j+\ell, k, k+\ell$ lie in this circular order. We leave this detail to the end.

Let $L_1, L_2$ be the length of the portions of $\eta'''$ between $y_{j}$ and $y_{j+\ell+1}$, and between $y_{k}$ and $y_{k+\ell+1}$, respectively. Suppose first that for at least one of the pairs $(y_j, y_k)$, $(y_{j+\ell+1}, y_{k+\ell+1})$, the two points are separated by distance larger than $\eps$. In this case, if we interchange the positions of $y_{j+\ell'}$ and $y_{k+\ell'}$ (which, as mentioned before, are separated by distance larger than $\eps$) for all $1\le \ell'\le \ell$, then $L_1+L_2$ decreases significantly (by an amount that does not tend to $0$ with $\eps$). This contradicts the assumption that $\eta'''$ is $\eps$-close to the shortest homotopic curve $\gamma'$.

Now suppose $(y_j, y_k)$ are $\eps$-close to each other, as well as $(y_{j+\ell+1}, y_{k+\ell+1})$. If $L_1$ is significantly different from $L_2$ (by an amount that does not tend to $0$ with $\eps$), say with $L_1>L_2$, then we could have moved each $y_{j+\ell'}$ to within $\eps$ of $y_{k+\ell'}$, lowering $L_1$ to $L_2$ and making $\eta'''$ shorter. And if $L_1$ is almost equal to $L_2$ then, letting $\eps\to 0$ we obtain a counterexample to the uniqueness of the shortest homotopic path from $y_j$ to $y_{j+\ell+1}$. This finishes our mix-and-match argument.

Now we prove that the indices $j, j+\ell, k, k+\ell$ lie in this circular order, as promised. Suppose for a contradiction that $j<k\le j+\ell$. Then the edges $e_{j+1}, \ldots, e_k$ are the same as the edges $e_{k+1}, \ldots, e_{2k-j}$, meaning, the curve $\eta'''$ winds twice in a row along the same edges. But then, by a result of Hershberger and Snoeyink~\cite{hs94}, which we call the \emph{two windings lemma} and include in Appendix~\ref{app_windings}, $y_{j+u}$ must be $\eps$-close to $y_{k+u}$ for some $1\le u\le k-j$. This is a contradiction, as mentioned above.

This concludes the proof of Theorem~\ref{thm_crossings} for the case of the number of self-intersections of a single curve. The case of the number of intersections of two curves is similar, though slightly simpler since there is no need to invoke the two windings lemma.
\end{proof}

\begin{reptheorem}{thm_bound_by_peeling}
For a fixed obstacle set $P$, the $P$-curve that maximizes the number of HCS iterations is the boundary of the convex hull of $P$.
\end{reptheorem}

\begin{proof}
Let $\gamma_0$ be a given $P$-curve, and let $\delta_0$ be the boundary of the convex hull of $P$. The curves $\gamma_0$ and $\delta_0$ are disjoinable, and furthermore, their separation into disjoint $\widehat\gamma_0$, $\widehat\delta_0$ can be done in such a way that $\widehat\gamma_0$ lies on the bounded side of $\widehat\delta_0$. We say for short that $\delta_0$ \emph{bounds} $\gamma_0$. Let $\gamma_{i+1} = \HCS_P(\gamma_i)$ and $\delta_{i+1} = \HCS_P(\delta_i)$ for all $i$. The proof of Theorem~\ref{thm_crossings} describes how to obtain disjoint $\widehat\gamma_{i+1}$, $\widehat\delta_{i+1}$ that are $\eps$-close to $\gamma_{i+1}$, $\delta_{i+1}$, from the corresponding curves $\widehat\gamma_i$, $\widehat\delta_i$. Furthermore, each of the steps described in the proof (namely the snapping, edge-sequence reduction, and vertex sliding steps) conserves the property that $\widehat\gamma_{i}$ remains on the bounded side of $\widehat\delta_i$. Hence, $\delta_i$ bounds $\gamma_i$ for all $i$. Therefore, the HCS process starting with $\gamma_0$ does not take more iterations than the one starting with $\delta_0$.
\end{proof}

\begin{reptheorem}{thm_gp}
Let $P$ be an obstacle set in general position. Let $\gamma$ be a simple $P$-curve. Then $\HCS_P(\gamma)$ is also simple. Let $\gamma_1, \gamma_2$ be disjoint $P$-curves. Then $\HCS_P(\gamma_1), \HCS_P(\gamma_2)$ are also disjoint. Hence, under HCS with obstacles in general position, a simple curve stays simple, and a pair of disjoint curves stay disjoint.
\end{reptheorem}

\begin{proof}
Let $P$ be in general position, let $\mathcal T$ be a triangulation of $P$, and let $\gamma$ be a simple $P$-curve. Suppose that $\gamma'=\HCS_P(\gamma)$ did not collapse to a point. We will show that $\gamma'$ is also simple. Let $\eta$, $\eta'$, $\eta''$, $\eta'''$ be as in the proof of Theorem~\ref{thm_crossings}; see again Table~\ref{table_curves}.

Let $(e_0, \ldots, e_{m-1})$ be the circular list of edges visited by $\eta''$ and $\eta'''$. For each $0\le i<m$, let $y_i$ be the vertex of $\eta''$ on $e_i$, and let $z_i$ be the vertex of $\eta'''$ on $e_i$. Hence, $\eta''$ was transformed into $\eta'''$ by sliding each $y_i$ to $z_i$ along $e_i$.

We already know that $\eta'''$ is simple, and so no two vertices $y_i$ on the same edge switched order during the sliding. Hence, all we need to show is that $\gamma'$ does not visit any vertex twice.

\begin{figure}
\centerline{\includegraphics{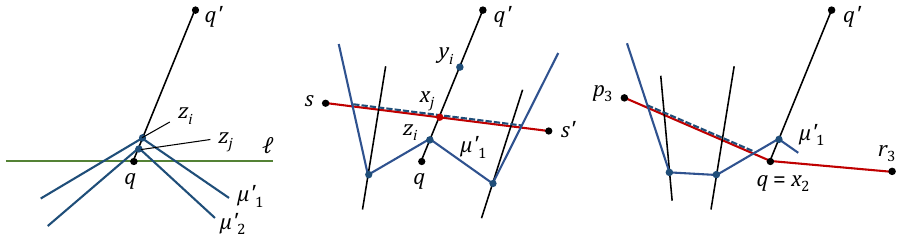}}
\caption{\label{fig_gp_tr}Left: Assuming for a contradiction that $\gamma'$ visits an obstacle $q$ twice. Center: The red segment $ss'$ is part of $\gamma$, and $\eta'''$ (blue) could be made shorter by taking the dashed shortcut. Right: The red segments $p_3qr_3$ are part of $\gamma$, and again $\eta'''$ (blue) could be made shorter by taking the dashed shortcut.}
\end{figure}

Suppose for a contradiction that $q\in P$ is visited twice by $\gamma'$, with the two visits being $\mu_1=p_1 q r_1$ and $\mu_2=p_2 q r_2$. Let $\mu'_1, \mu'_2$ be the portions of $\eta'''$ corresponding to $\mu_1, \mu_2$. Since $\eta'''$ is simple and $\eps$-close to shortest, there must exist a line $\ell$ through $q$ such that $p_1,r_1,p_2,r_2$ all lie on one side of $\ell$, and such that $\mu'_1$, $\mu'_2$ briefly cross to the other side of $\ell$ when going around $p$, before crossing back. See Figure~\ref{fig_gp_tr} (left). Assume for simplicity that $\ell$ is horizontal, with $p_1,r_1,p_2,r_2$ below $\ell$. Let $e=qq'$ be a triangulation edge with $q'$ above $\ell$. Let $z_i,z_j$ be the vertices of $\eta'''$ in $\mu'_1\cap e$, $\mu'_2\cap e$, respectively, where $e_i=e_j=e$. Assume without loss of generality that $z_i$ lies higher than $z_j$.

The corresponding vertices $y_i,y_j$ of $\eta''$ are $\eps$-close to two points $x_i,x_j\in \gamma\cap e$, where $x_i$ lies higher than $x_j$. The point $x_j$ could be either the lower endpoint $q$ or somewhere in the middle of $e$ (it cannot be the upper endpoint $q'$ since then we would also need to have $x_i=q$, so $\gamma$ would self-intersect).

Suppose first that $x_j$ is somewhere in the middle of $e$, so $x_j\in ss'$ for some $s,s'\in P$. The segment $ss'$ cannot equal $e$, since then $\gamma$ would self-intersect as before. Hence, $y_i$ crossed $ss'$ on its way to $z_i$. Let $i'\le i$ and $i''\ge i$ be minimal and maximal indices such that all the edges $e_{i'}, \ldots, e_{i''}$ intersect $ss'$ and the corresponding vertices $y_{i'},\ldots,y_{i''}$ crossed $ss'$ on their way to $z_{i'},\ldots,z_{i''}$, in the transition from $\eta''$ to $\eta'''$. (The edges $e_{i'}, \ldots, e_{i''}$ cannot be \emph{all} the edges $e_0, \ldots, e_{m-1}$, since then $\eta''$ would be collapsible to a point.)
But then $\eta'''$ could be made significantly shorter (by an amount that does not tend to $0$ as $\eps\to 0$) by not having these vertices cross $ss'$. See Figure~\ref{fig_gp_tr} (center). Contradiction.

Now suppose $x_j=q$, so $x_j$ is part of a visit $p_3 q r_3$ of $\gamma$. The edge $e$ must be on the side of at which the angle $p_3qr_3$ is less than $180$ degrees. Hence, one of the points $p_3, r_3$ must be above $\ell$. Say it is $p_3$. Hence, the portion $\mu'_1$ of $\eta'''$ intersects $p_3q$, so so one or more consecutive vertices of $\eta''$ moved across $p_3q$ in the transition from $\eta''$ to $\eta'''$. Thus, $\eta'''$ could be made shorter by not crossing $p_3q$, as before. Contradiction.

The case of two curves is similar.
\end{proof}

\subsection{Proof of Theorem~\ref{thm_infl}}

In this section we prove Theorem~\ref{thm_infl}, regarding the number of inflection edges. The proof is based of the vertex-release algorithm. The basic idea is that, given a self-disjoinable curve, if the vertex releases are performed in an appropriate order, then the curve stays self-disjoinable at all times. Moreover, no vertex release increases the number of inflection edges. Along the way, we develop enough machinery to re-prove Theorem~\ref{thm_crossings}.

We will use the type-2 curves introduced above. We also introduce two other types of curves that are arbitrarily close to $P$-curves. We call them \emph{type-1} and \emph{type-3} curves.

\begin{figure}
\centerline{\includegraphics{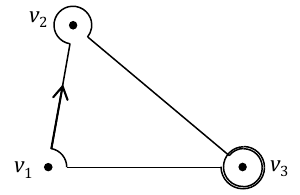}}
\caption{\label{fig_type1}Type-1 curve with $\alpha_1 = 80^\circ, \alpha_2=-300^\circ, \alpha_3=-680^\circ$.}
\end{figure}

\paragraph{Type-1 curves.}
Let $P$ be a set of obstacle points, let $\gamma$ be a $P$-curve that goes through obstacles $(p_0, \ldots, p_{k-1})$ in this circular order, and let $\eps>0$ be small enough. For each $p\in P$, let $C_p$ be a circle of radius $\eps$ centered at $p$. A \emph{type-1} curve $\delta$ corresponding to $\gamma$ is composed of \emph{straight parts} and \emph{circular parts}. For each $i$, let $e_i$ be the segment $p_{i-1}p_i$. Each straight part goes from the point $y_{i-1}=e_i\cap C_{p_{i-1}}$ to the point $x_i=e_i\cap C_{p_i}$. And each circular part goes along $C_{p_i}$ from $x_i$ to the next point $y_i$, either clockwise or counterclockwise, describing any number of turns around $C_i$.\footnote{The circular parts can be approximated arbitrarily closely by piecewise-linear paths.} Hence, $\delta$ can be combinatorially specified by associating to each $p_i$ a signed angle $\alpha_i$ that is congruent modulo $2\pi$ to $\angle p_{i-1}p_ip_{i+1}$. See Figure~\ref{fig_type1}. The unique $P$-curve corresponding to a given type-1 curve $\delta$ is denoted by $\tP(\delta)$.

\paragraph{Type-3 curves.} A \emph{type-3} curve is an $(\eps^2)$-perturbation of a type-1 curve. A type-3 curve is a self-transversal curve composed of \emph{straight parts} and \emph{spiral parts}. Each spiral part is centered at some $p\in P$, and its initial and final radii satisfy $\eps<r_0<r_1<\eps+\eps^2$. The interval $[r_0, r_1]$ is called the \emph{radial interval} of the spiral part. Different spiral parts centered at the same obstacle $p$ have disjoint radial intervals. Furthermore, the endpoints of the spiral parts are displaced either clockwise or counterclowise by distance up to $\eps^2$ from the corresponding endpoints in the type-1 curve. It does not matter whether the spiral parts spiral out clockwise or counterclockwise, since the number of self-intersections is unaffected.

Hence, a type-3 curve can be specified purely combinatorially, by specifying, for each obstacle $p\in P$, the relative order of the radial intervals of the spiral parts around $p$, and for each other obstacle $q\neq p$, the clockwise order of the endpoints near $p$ of the straight parts that go between $p$ and $q$. See Figure~\ref{fig_type3} for an example.

\begin{figure}
\centerline{\includegraphics{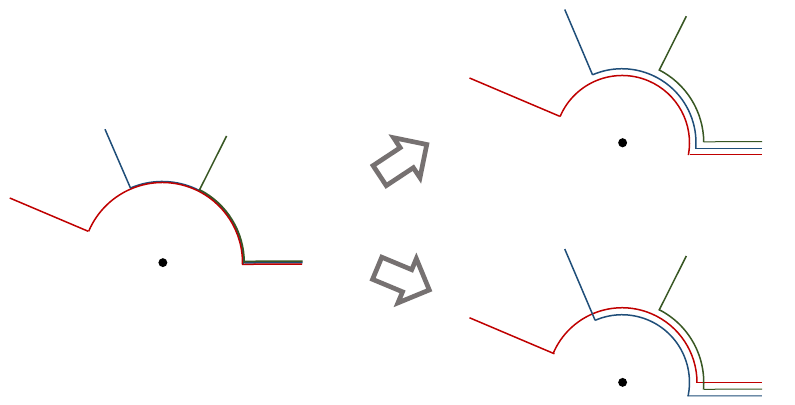}}
\caption{\label{fig_type3}Two different type-3 realizations of the same type-1 curve. The type-1 curve in this example visits the shown obstacle three times. The three visits are shown in different colors.}
\end{figure}

More generally, we define a \emph{collection of type-3 curves} on the same obstacle set, and we specify their relation to one another purely combinatorially in a similar way.

Each type-3 curve $\zeta$ corresponds (is $(\eps^2)$-close) to a unique type-1 curve, which we denote by $\tOne(\zeta)$. We also define $\tP(\zeta)$ as $\tP(\tOne(\zeta))$.

\begin{observation}\label{obs_2to3}
Let $\delta$ be a self-transversal type-2 curve. Then $\delta$ can be turned into a self-transversal type-3 curve $\zeta$ $\eps$-close to it, such that their number of self-intersections satisfy $\cros(\zeta)= \cros(\delta)$. Similarly, let $\delta_1, \delta_2$ be transversal type-2 curves. Then they can be turned into transversal type-3 curves $\zeta_1, \zeta_2$ $\eps$-close to them whose number of intersections satisfy $\cros(\zeta_1,\zeta_2)=\cros(\delta_1,\delta_2)$.
\end{observation}

\begin{proof}
Assume without loss of generality that the given type-2 curves do not pass through any $p\in P$. We turn the short parts of the type-2 curves into spiral parts of type-3 curves. The shorter the part, the larger the radius. See Figure~\ref{fig_2to3}. The number of intersections stays unchanged.
\end{proof}

\begin{figure}
\centerline{\includegraphics{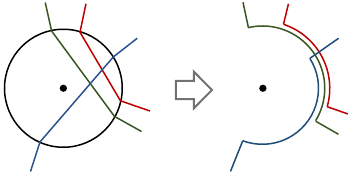}}
\caption{\label{fig_2to3}Turning a type-2 curve into a type-3 curve.}
\end{figure}

Let $\gamma$ be a $P$-curve, let $\delta$ be a type-2 curve corresponding to $\gamma$, and let $\zeta$ be the type-3 curve obtained from $\delta$ as in Observation~\ref{obs_2to3}. For every nailed visit that $\gamma$ makes to an obstacle $p\in P$, the corresponding type-3 curve $\zeta$ has a corresponding visit to $p$ with a spiral part that describes an angle very close to $\pi$. We call this visit of $\zeta$ \emph{nailed}. Call a type-3 curve \emph{steady} if, for every $p\in P$, the nailed spiral parts around $P$ have smaller radii than the non-nailed ones. Then our curve $\zeta$ is steady by construction.

A visit of a type-3 curve to an obstacle $p\in P$ is called \emph{stable} (resp.~\emph{unstable}) if the visit of the corresponding type-1 curve to $p$ is stable (resp.~unstable).

In order to release an unstable visit to an obstacle $p$ in a type-3 curve, we proceed as in the type-1 curve it realizes, and then we decide on the radius and endpoint order of the spiral parts that were newly created or modified. If $p$ is both preceded and followed by the same obstacle $q$, recall that we merge the two visits of $q$. If at least one of these visits was nailed, then we mark the new visit as nailed.

\begin{lemma}\label{lemma_release_type3}
Let $\gamma$ be a steady type-3 curve with at least one unstable, non-nailed obstacle visit. Then there is a way to perform a vertex release on one such visit, such that the number of self-intersections does not increase, and such that the resulting curve is also steady.
\end{lemma}

\begin{proof}
Let $p$ be an obstacle that has at least one unstable non-nailed visit. Release, from among all unstable visits to $p$, the one $v_i$ with the largest radius in $\gamma$. Suppose first that $v_{i-1}$ and $v_{i+1}$ visit different obstacles. In each new obstacle visit $w_k$, we give the new spiral part the largest radius around in that obstacle. The relative order of the starting and ending points of that spiral part are chosen so that the spiral part is as short as possible. The spiral part of $v_{i-1}$ has an endpoint $a$ that stays in place and an endpoint $b$ that moves. The new position of the endpoint $b$ is chosen as close as possible to its old position. The same is done for $v_{i+1}$.

Hence, a piece $\gamma_{\mathrm{old}}$ of $\gamma$ is replaced by another piece $\gamma_{\mathrm{new}}$, and $\gamma_{\mathrm{old}}, \gamma_{\mathrm{new}}$ together form a simple curve, bounding a region $R$. See Figure~\ref{fig_release_type3} (left). The number of self-intersections does not increase, since any other curve piece that enters $R$ by crossing $\gamma_{\mathrm{new}}$, must exit $R$ by crossing $\gamma_{\mathrm{old}}$. 

\begin{figure}
\centerline{\includegraphics{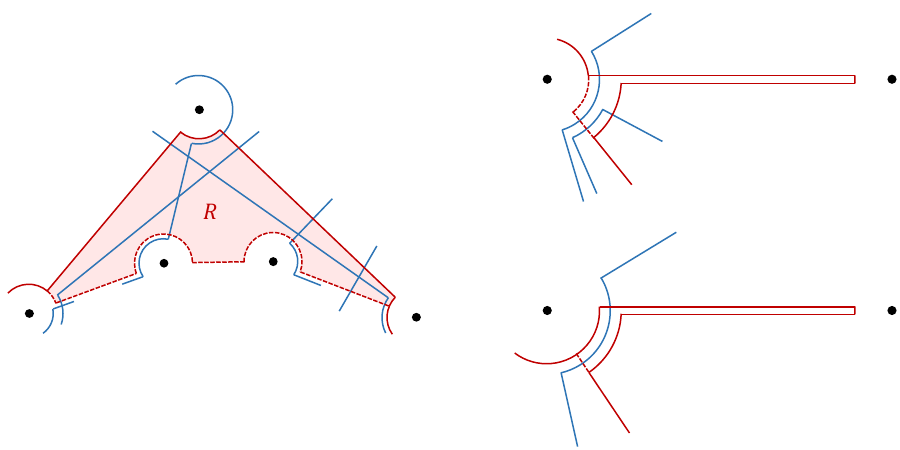}}
\caption{\label{fig_release_type3}Vertex releases on type-3 curves. In each example, the red curve undergoes a vertex release. The dotted curves are the resulting curves after the release. The blue curves are examples of other curves (of portions thereof) present in the vicinity. As a case analysis shows, whenever a blue curve intersects the red curve after the release, the two curves also intersected before the release.}
\end{figure}

If $v_{i-1}$ and $v_{i+1}$ visit the same obstacle $q$, then there are two spiral parts $c_1, c_2$ around $q$ that are merged into one. We give the new spiral part the smaller radius among the old radii of $c_1$, $c_2$. This way, the resulting curve is also steady. Further, a case analysis shows that the number of self-intersections does not increase. See Figure~\ref{fig_release_type3} (right).
\end{proof}

Lemma~\ref{lemma_release_type3} provides an alternative proof of Theorem~\ref{thm_crossings}: We repeatedly release non-nailed unstable vertices according to the lemma, until no such vertices remain.

\begin{lemma}\label{lemma_turn}
Let $C$ be a circle, let $p,r$ be points on $C$, and let $q$ be a point in the interior of the disk bounded by $C$. Suppose the ordered triple $p,q,r$ describes a clockwise (resp.~counterclockwise) turn. Let $\gamma$ be a simple piecewise-linear path that goes from $p$ through $q$ to $r$ and does not otherwise intersect $C$. Then $\gamma$ has at least one clockwise (resp.~counterclockwise) vertex. See Figure~\ref{fig_turn} (left).
\end{lemma}

\begin{figure}
\centerline{\includegraphics{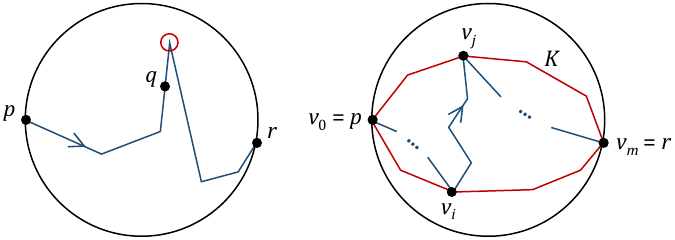}}
\caption{\label{fig_turn}Left: A simple path from $p$ through $q$ to $r$ that does not leave the circle must contain at least one clockwise vertex (marked in red). Right: Proof of the claim: The path must turn clockwise at $v_j$.}
\end{figure}

\begin{proof}
Let $K$ be the boundary of the convex hull of $\gamma$. Let $v_0, \ldots, v_{n-1}$ be the vertices of $K$ in counterclockwise order, with $v_0=p$. Let $m$ be such that $v_m=r$. Let $v_j$ be the first vertex with $m<j<n$ visited by $\gamma$, and let $v_i$ be the last vertex with $0\le i<m$ visited by $\gamma$ before $v_j$. Then the subpath of $\gamma$ from $v_i$ to $v_j$, together with the part of $K$ counterclockwise from $v_i$ to $v_j$, form a simple closed curve, which $\gamma$ cannot cross. Since $\gamma$ must end at $r=v_m$, it must turn clockwise at $v_j$. See Figure~\ref{fig_turn} (right).
\end{proof}

Let $\gamma$ be a self-disjoinable $P$-curve, and let $\delta$ be a simple type-3 curve realizing $\gamma$. (We know such a $\delta$ exists by Lemma~\ref{lemma_snap} and Observation~\ref{obs_2to3}.) Let $(p_0, \ldots, p_k)$ be the minimal sequence of vertices of $\gamma$ (i.e.~this sequence omits obstacles at which $\gamma$ continues in a straight line). For each $i$, let $e_i$ be the edge $e_i=p_{i-1}p_i$. We will define the notion of an edge $e_i$ \emph{being realized by $\delta$ as an inflection edge}. For this, we first define the notion of a vertex $p_i$ \emph{being realized by $\delta$ as a clockwise (counterclockwise) turn}. If $p_{i+1}$ lies to the right (resp.~left) of the directed edge $e_i$, then we say that $p_i$ is realized by $\delta$ as a clockwise (resp.~counterclockwise) turn, irrespective of $\delta$. Now suppose $p_{i+1}$ lies on the ray emanating from $p_i$ through $p_{i-1}$ (so $\gamma$ makes a U-turn at $p_i$). Let $q_i$ be the obstacle point visited by $\gamma$ right before and after $p_i$ (note that $q_i$ is not necessarily a vertex of $\gamma$). Let $L_i$, $L'_i$ be the long parts of $\delta$ corresponding to the segments $q_ip_i$, $p_iq_i$ of $\gamma$, respectively. If $L'_i$ lies to the right of $L_i$ (considering $L_i$ as a directed segment), then we say that $p_i$ is realized by $\delta$ as a clockwise turn. Otherwise, we say that $p_i$ is realized by $\delta$ as a counterclockwise turn. The turning direction of the spiral parts of $\delta$ are irrelevant for this definition. Finally, we say that an edge $e_i$ is realized by $\delta$ as an inflection edge if one of $p_{i-1}, p_i$ is realized by $\delta$ as a clockwise turn and the other one as a counterclockwise turn. Denote by $\infl_\gamma(\delta)$ the number of edges of $\gamma$ that are realized by $\delta$ as inflection edges.

For every obstacle $q$ at which $\gamma$ continues in a straight line, we say that $\delta$ is \emph{straight} at $q$.

\begin{reptheorem}{thm_infl}
Let $\gamma$ be self-disjoinable, and let $\gamma'=\HCS_P(\gamma)$. Then their inflection-edge numbers satisfy $\infl(\gamma')\le \infl(\gamma)$. Hence, under HCS on a self-disjoinable curve, the curve's number of inflection edges never increases.
\end{reptheorem}

\begin{proof}
Let $\gamma$ be a self-disjoinable $P$-curve, and let $\gamma'=\HCS_P(\gamma)$. Let $\eps>0$ be small enough. Let $\widehat\gamma$ be a simple piecewise-linear curve $(\eps^2)$-close to $\gamma$ minimizing the number of inflection edges $\infl(\widehat\gamma)$. We will construct a simple piecewise-linear curve $\rho$ that is $\eps$-close to $\gamma'$ and satisfies $\infl(\widehat\gamma)\ge\infl(\rho)$.

Let $\delta=\snap(\widehat\gamma)$. By Lemma~\ref{lemma_snap}, $\delta$ is also simple. Let $\zeta$ be the type-3 curve obtained from $\delta$ according to the procedure in Observation~\ref{obs_2to3}. Then $\zeta$ is simple as well. Consider a vertex $p_i$ of $\gamma$ that is realized by $\zeta$ as a clockwise turn. By the construction of $\delta$ and $\zeta$ and by Lemma~\ref{lemma_turn}, $\widehat\gamma$ must have a clockwise vertex inside the circle $C_{p_i}$. Similarly, if $p_i$ is realized by $\zeta$ as a counterclockwise turn, then $\widehat\gamma$ must have a counterclockwise vertex inside $C_{p_i}$. Hence, $\infl(\widehat\gamma)\ge\infl_\gamma(\zeta)$.

Now repeatedly release non-nailed unstable obstacle visits from $\zeta$ as in Lemma~\ref{lemma_release_type3}, until none remain, obtaining $\zeta'$. Let $\zeta_0, \ldots, \zeta_k$ be the sequence of type-3 curves produced in this process, with $\zeta_0=\zeta$ and $\zeta_k=\zeta'$. Then $\zeta_i$ is simple for each $i$. For each $i$, let $\gamma_i = \tP(\zeta_i)$. In particular, $\gamma'=\gamma_k$.

We claim that no vertex release increases the number of inflection edges, meaning $\infl_{\gamma_i}(\zeta_i)\ge\infl_{\gamma_{i+1}}(\zeta_{i+1})$ for all $i$.

Indeed, consider the release of vertex $p_j$ of $\gamma_i$. Say it is realized by $\zeta_i$ as a clockwise turn. Suppose first that $\gamma_i$ does not make a U-turn at $p_j$. Then $p_j$ is replaced in $\gamma_{i+1}$ by zero or more new vertices, all of which are realized by $\zeta_{i+1}$ as clockwise turns. Let $q, q'$ be the obstacles crossed by $\gamma_i$ just before and after $p_j$, respectively. Then each of $q,q'$ could change from counterclockwise to straight, or from counterclockwise to clockwise, or from straight to clockwise, or they could stay as they are. This is true even if one of $\gamma_i$, $\gamma_{i+1}$ makes a U-turn at one or both of these vertices. Hence, the number of alternations between clockwise and counterclockwise did not increase at step~$i$.

\begin{figure}
\centerline{\includegraphics{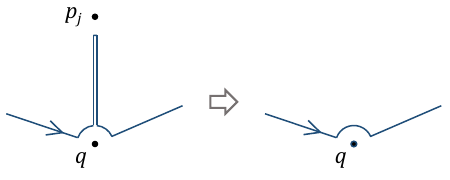}}
\caption{\label{fig_release_cw}When releasing a vertex $p_j$ that makes a clockwise U-turn, if after the release the curve makes a counterclockwise turn at $q$, then before the release both visits to $q$ were counterclockwise.}
\end{figure}

Now suppose that $\gamma_i$ makes a U-turn at $p_j$. Let $q$ be the obstacle that precedes and follows $p_j$ in $\gamma$. Recall that in this case, the visit to $p_j$ is removed and the two visits to $q$ are merged into one. As before, here there are several options as to whether the visit(s) of $q$ by $\gamma_i$, $\gamma_{i+1}$ are straight, or are realized as counterclockwise or clockwise by $\zeta_i, \zeta_{i+1}$. However, if the visit of $q$ by $\gamma_{i+1}$ is realized as counterclockwise, then both visits by $\gamma_i$ had to be realized as counterclockwise. See Figure~\ref{fig_release_cw}. Hence, in this case as well, the number of alternations between clockwise and counterclockwise did not increase at step $i$. This finishes the proof that $\infl_{\gamma_i}(\zeta_i)\ge\infl_{\gamma_{i+1}}(\zeta_{i+1})$ for each $i$. Hence, $\infl_\gamma(\zeta)\ge\infl_{\gamma'}(\zeta')$.

\begin{figure}
\centerline{\includegraphics{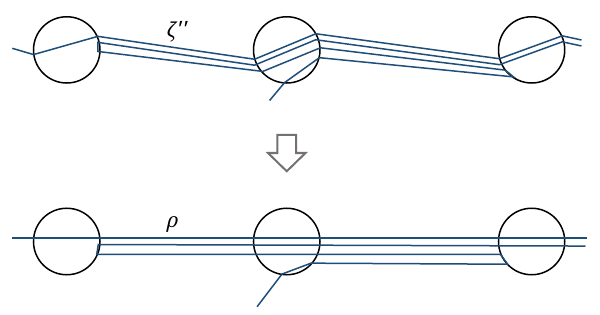}}
\caption{\label{fig_straighten}Removing slight inflection edges.}
\end{figure}

Next, we turn $\zeta'$ into a type-2 curve $\zeta''$ by reversing the procedure of Observation~\ref{obs_2to3}; namely, we turn the spiral parts of $\zeta'$ into straight segments, and we move the endpoints of these segments to distance $\eps$ of the corresponding obstacles. This does not introduce any self-intersections, so $\zeta''$ is simple. There is a slight final problem: For each obstacle at which $\gamma'$ goes straight, the corresponding short part of $\zeta''$ might be an inflection edge. This problem is solved by straightening the appropriate portions of $\zeta''$ as in Figure~\ref{fig_straighten}. Let $\rho$ be the straightened curve. Then $\rho$ is still simple, and no short part of $\rho$ is an inflection edge. Hence, $\infl_{\gamma'}(\zeta')\ge\infl(\rho)$. Summing up, we have
\begin{equation*}
\infl(\widehat\gamma)\ge\infl_\gamma(\zeta)\ge\infl_{\gamma'}(\zeta')\ge\infl(\rho),
\end{equation*}
concluding the proof.
\end{proof}

\section{Discussion}

One of the reasons continuous curve-shortening flows were introduced and studied was to overcome the shortcomings of the \emph{Birkhoff curve-shortening process} (\cite{birkhoff}, see also e.g.~\cite{croke}), specifically the fact that it might cause the number of curve intersections to increase~\cite{grayson_embedded,hass_scott}. As we have shown, HCS is a discrete process that overcomes this flaw without introducing analytical difficulties, at least in the plane. We conjecture that the definition of HCS can be extended without modifications to more general surfaces. It would be interesting to check whether such an extension has the properties which HCS has in the plane. The answer might depend on the properties of the surface.

\paragraph{Acknowledgements.} Thanks to Arseniy Akopyan, Imre B\'ar\'any, Jeff Erickson, Radoslav Fulek, Jeremy Schiff, Arkadiy Skopenkov, and Peter Synak for useful discussions. Thanks also to the referees for their useful comments.

\bibliographystyle{plainurl}
\bibliography{homotopic_arxiv}

\appendix

\section{Uniqueness of shortest homotopic curve}\label{app_unique}

The following lemma is included for completeness.

\begin{lemma}\label{lemma_unique_shortest}
Let $P$ be a set of polygonal and point obstacles, and let $\gamma$ be a path or curve that avoids $P$ (except possibly at the endpoints). Then there exists a unique path or curve $\delta$ homotopic to $\gamma$ that is locally shortest.
\end{lemma}

\begin{proof}
The basic idea is that length function on the space of piecewise-linear curves is convex.

We use the triangulation approach of Section~\ref{subsec_triang}. Suppose for a contradiction that $\delta_1$ and $\delta_2$ are two different locally shortest curves or paths homotopic to $\gamma$. Clearly, their edge sequences must be reduced, so we have $\mathcal E(\delta_1)=\mathcal E(\delta_2)$. Denote this common edge sequence by $\mathcal E$. All vertices of $\delta_1, \delta_2$ must lie on edges of $\mathcal E$. For each $e\in\mathcal E$, let $x_{e,1}, x_{e,2}$ be the vertex of $\delta_1, \delta_2$ on $e$, respectively. Say the lengths of the curves satisfy $|\delta_1|\ge |\delta_2|$.

\begin{figure}
\centerline{\includegraphics{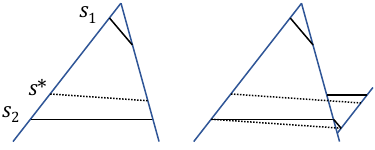}}
\caption{\label{fig_midsegment}Proof of an inequality involving the lengths of three segments.}
\end{figure}

Given $\eps>0$, for each $e\in\mathcal E$ define the point $x^*_e = (1-\eps) x_{e,1} + \eps x_{e,2}$. Let $\delta^*$ be the curve or path that uses the sequence of points $x^*_e$ as vertices. Clearly, $\delta^*$ is homotopic to $\gamma$. For each two adjacent edges $e_1, e_2\in \mathcal E$, the corresponding curve segments $s_1, s_2, s^*$ satisfy $|s^*| + (\eps/(1-\eps))|s^*| \le |s_1| + (\eps/(1-\eps))|s_2|$ (see Figure~\ref{fig_midsegment}), which implies $|s^*|\le (1-\eps)|s_1| + \eps|s_2|$. Hence, $|\delta^*| \le (1-\eps)|\delta_1|+\eps|\delta_2| \le |\delta_1|$. This is true for all $\eps>0$, contradicting the assumption that $\delta_1$ is locally shortest.
\end{proof}

\section{The two windings lemma}\label{app_windings}

The following fact that was stated and proven in \cite{hs94} in a somewhat different context. We include the proof for completeness.

\begin{lemma}
Let $P$ be a finite set of obstacle points. Let $\mathcal T$ be a triangulation of $P$. Let $\gamma$ be a piecewise-linear curve that avoids $P$, such that its edge sequence $\mathcal E(\gamma)=(e_0, \ldots, e_{m-1})$ is reduced (meaning, $e_i\neq e_{i+1}$ for all $i$). Suppose that there exists $k\le m/2$ such that $e_i = e_{i+k}$ for all $0\le i<k$. (Meaning, $\gamma$ winds twice in a row along a certain sequence of edges, and then does some other stuff before returning to its starting point.)

Let $\delta$ be the shortest curve homotopic to $\gamma$. Let $(x_0, \ldots, x_{m-1})$ be the vertices of $\delta$, with $x_i\in e_i$ for each $0\le i<m$. Then there exists some $0\le i <k$ for which $x_i=x_{i+k}$.
\end{lemma}

\begin{figure}
\centerline{\includegraphics{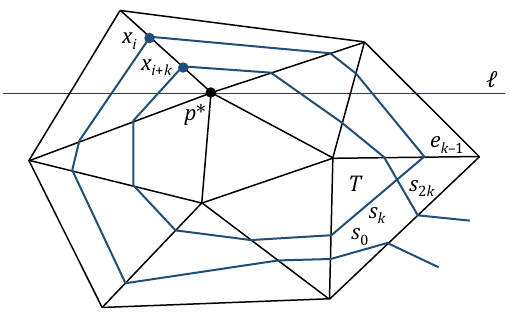}}
\caption{\label{fig_two_windings} Proof of the two windings lemma.}
\end{figure}

\begin{proof}
For each $i$, let $s_i$ be the segment $s_i=x_{i-1}x_i$ of $\delta$. Note that the three segments $s_0,s_k,s_{2k}$ lie in the same triangle $T\in\mathcal T$. Suppose without loss of generality that the edge $e_{k-1} = e_{2k-1}$ is horizontal and $T$ lies below it. For each $i$, let $p_i\in P$ be the lower endpoint of the edge $e_i$ (breaking ties arbitrarily). Let $p^*$ be the highest point among $\{p_0, \ldots, p_{k-1}\}$. Let $\ell$ be a horizontal line through $p^*$. Then the part of $\gamma$ between $e_0$ and $e_{k-1}$ crosses $\ell$ upwards, then crosses at least one non-horizontal edge $e_i$ that has $p_i=p^*$, and then crosses $\ell''$ downwards. Hence, $\gamma$ can be shortened by not rising above $\ell''$, and therefore, in $\delta$ we have $x_i = p_i=p^*$. The same argument yields $x_{i+k} = p_{i+k}=p^*$. See Figure~\ref{fig_two_windings}.
\end{proof}

\section{Implementation details}\label{app_implementation}

Following~\cite{gp_acsf}, we simulated ACSF on the curve $\Delta$ of Section~\ref{subsec_experiments} using a simple front-tracking approach: We initially sampled $m=1000$ points uniformly spaced along $\Delta$. For each such point $p_i$, we estimate its normal vector and radius of curvature by the normal vector $v_i$ and radius $r_i$ of the unique circle passing through points $p_{i-1}, p_i, p_{i+1}$. We simultaneously let all points move at the appropriate speeds for a short time interval $t=\min{\{c\cdot (d_{\mathrm{min}})^{4/3}, t_{\mathrm{min}}\}}$, where $d_{\mathrm{min}}$ is the minimum distance between two consecutive points, and $c=3\cdot 10^{-4}$ and $t_{\mathrm{min}}=3\cdot 10^{-9}$ are fixed parameters. Hence, as the minimum distance between points decreases, we take smaller time steps, except that we fix a minimum time step $t_{\mathrm{min}}$ in order to be able to go pass the singularity at which the self-intersection disappears. In order to prevent the points $p_i$ from bunching together at sharp bends of the curve, each point $p_i$ is also given a tangential velocity that tends to move it away from its closer neighbor among $p_{i-1},p_{i+1}$. These tangential velocities should not affect the evolution of the flow, since they only cause the curve points to move within the curve.

We simulated HCS on uniform-grid obstacles and on random obstacles using two different programs. Both programs use the vertex-release method described in Section~\ref{sec_vrelease}. The random-obstacle program stores points in a quadtree, in order to be able to answer triangle-containment queries efficiently. The random-obstacle program does not handle nailed obstacles, since in theory they occur with probability zero.

The uniform-grid program is more memory-efficient than the random-obstacle program, since it does not need to store all obstacles in memory, but rather takes the obstacle set to be $\Z^2$. Vertex releases on the curve $\gamma$ are performed as follows: Let $p,q,r$ be three consecutive obstacles along $\gamma$, where $q$ is to be released, and assume these three points do not lie on the same line. Note that the vectors $q-p$, $r-q$ are primitive (i.e.~each one has relatively prime coordinates). Let $z_1, \ldots, z_k$ be the points that should replace $q$ after the vertex release. Denote $z_0=p$, $z_{k+1}=r$. Then each of the triangles $z_i q z_{i+1}$ has area $1/2$. Hence, given $z_i$ and $q$, we can find $z_{i+1}$ by an application of the extended Euclidean gcd algorithm, plus some additional calculations. We leave the details to the reader.

Our code is available at the following links:

\url{https://github.com/savvakumov/ACSF-simulation}

\url{https://github.com/savvakumov/HCS-with-random-obstacles}

\url{https://github.com/savvakumov/HCS-square-grid}

\end{document}